\documentclass[preprint,12pt]{elsarticle}

\usepackage{amssymb} 
\usepackage{bm}
\usepackage{dsfont} 
\usepackage{mathtools} 

\usepackage{tikz}
\usetikzlibrary{patterns, arrows.meta, decorations, decorations.pathmorphing, calc}

\newcommand{\valvula}[4]{
    \begin{scope}[shift={(#1,#2)}]
        \draw[black, thick, fill=white] (0.5, -0.3) -- (-0.5, 0.3) -- (-0.5, -0.3) -- (0.5, 0.3) -- cycle;
        
        \draw[black, thick] (0,0) -- (0, 0.5) -- (0.3, 0.5) -- (-0.3, 0.5) arc[start angle=180, end angle=0, radius=0.3cm];
        
        \draw[black, thick, ->] (-0.5, -0.5) -- (0, -0.5) node[below]{#3} -- (0.5, -0.5);

        \node[above] at (0, 0.75) {#4};
    \end{scope}
}

\newcommand{\valvulades}[4]{
    \begin{scope}[shift={(#1,#2)}]
        \draw[black, thick, fill=black] (0.5, -0.3) -- (-0.5, 0.3) -- (-0.5, -0.3) -- (0.5, 0.3) -- cycle;
        
        \draw[black, thick, ->] (-0.5, -0.5) -- (0, -0.5) node[below]{#3} -- (0.5, -0.5);

        \node[above] at (0, 0.25) {#4};
    \end{scope}
}

\usepackage{hyperref} 

\usepackage{amsthm}

\mathtoolsset{showonlyrefs}

\hypersetup{
    pdftitle={A local RLS approach for discrete-time adaptive fuzzy controllers},
    pdfauthor={Víctor Costa da Silva Campos},
    pdfsubject={Adaptive Control Research Paper},
    pdfkeywords={Adaptive Control, Takagi-Sugeno Fuzzy Systems, Local Recursive Least Squares},
    colorlinks=true,
    linkcolor=black,
    citecolor=blue,
    urlcolor=blue
}

\renewcommand{\vec}[1]{\boldsymbol{#1}} 
\newcommand{\diag}[1]{\operatorname{\mathbf{diag}}\left(#1\right)} 
\theoremstyle{remark}
\newtheorem{asm}{Assumption}
\newtheorem*{rem}{Remark}
\theoremstyle{plain}
\newtheorem{thm}{Theorem}
\newtheorem{lem}{Lemma}

\journal{Journal of the Franklin Institute}

\begin{document}

\begin{frontmatter}



\title{A local recursive least squares approach for discrete-time adaptive fuzzy control}


\author[UFMG]{Víctor Costa da Silva Campos}

\affiliation[UFMG]{organization={Electronics Engineering Department, Universidade Federal de Minas Gerais},
            addressline={Av. Antônio Carlos, 6627, Pampulha},
            city={Belo Horizonte},
            postcode={31140-280}, 
            state={Minas Gerais},
            country={Brazil}}

\author[IFMG]{Mariella Maia Quadros}

\affiliation[IFMG]{organization={Area of Control and Industrial Processes, Instituto Federal de Educação, Ciência e Tecnologia de Minas Gerais, Campus Sabará},
            addressline={Rodovia MGC 262, km 10, Sobradinho},
            city={Sabará},
            postcode={34590-390}, 
            state={Minas Gerais},
            country={Brazil}}

\begin{abstract}
This paper proposes a local recursive least squares (RLS) estimation strategy for discrete-time adaptive fuzzy control of nonlinear systems represented in quasi-Linear Parameter Varying (qLPV)/Takagi--Sugeno (TS) form. Unknown nonlinear terms are approximated by a constant-consequent TS fuzzy model, and the consequent parameters are updated by a membership-function-weighted RLS law with a forgetting factor. The proposed estimator keeps a different covariance for each rule, considerably reducing the memory footprint of the least-squares updates. The adaptation is also simplified since each rule is only adapted when it is active. From these properties, we are capable of showing that the adaptation law ensures bounded local adaptation errors. Building upon this adaptation law, Linear Matrix Inequality (LMI) synthesis conditions are presented for matched, sector-bounded and norm-bounded unknown nonlinearities, guaranteeing ultimate uniform boundedness of the adaptive control system in closed loop. Three numerical examples are presented to illustrate the adaptive control conditions in the three cases: a planar manipulator with unknown gravity direction, a two-tank system with unknown coupling, and a \emph{Brushless DC} (BLDC) motor acting as thrust for an efficiency vehicle.
\end{abstract}


\begin{highlights}
    \item A membership-weighted local recursive least squares with forgetting factor approach is proposed for fuzzy TS models approximating unknown nonlinearities in nonlinear systems.
    \item LMI synthesis conditions are presented that ensure the ultimate uniform boundedness of the closed loop adaptive system (obtained from the proposed local adaptation laws) concerning the approximation error for matched, sector-bounded, and norm-bounded unknown nonlinearities. In the norm-bounded case, a separate feedforward condition is presented for a design that \emph{approximately} renders a desired output immune to it.
\end{highlights}


\end{frontmatter}


\section{Introduction}

Adaptive control aims to develop control laws for uncertain systems with controllers capable of adapting while ensuring closed loop stability and performance \citep{Tao2014,Annaswamy2021}. Its methods can usually be divided into direct adaptive control, which directly estimates the controller gains \citep{Fu2008,Yang2015}, and indirect adaptive control, which estimates unknown parameters or dynamics that are then used in the control law \citep{Ge2009,Zhang2024}.

In continuous-time adaptive control, the adaptation laws usually arise from the parameter adaptation errors in the time-derivative of the Lyapunov function, from regression errors, or from a combination of both (leading to composite adaptive control laws) \citep{Ortega2020,Guo2023}. Due to the fact that in discrete-time the parameter adaptation errors no longer appear linearly in the Lyapunov function analysis, discrete-time adaptation laws are more commonly derived from regression errors, and discrete-time adaptive control cannot be seen as a simple extension of continuous-time adaptive control \citep{Akhtar2004}.

One standard approach to estimating unknown nonlinearities is to employ universal approximators in the adaptive control loop \citep{Farrel2006}, such as neural networks \citep{Yang2015} and fuzzy systems \citep{Sun2024, Zhang2024}, since they allow unknown functions to be replaced by a set of unknown parameters and an approximation error (which can be made arbitrarily small over a compact operating region). In this regard, Takagi-Sugeno fuzzy systems \citep{Takagi1985} can be quite interesting since, given a set of fixed membership functions and consequents that are linear in the adjustable parameters, the resulting model will be linear in these adjustable parameters. In addition to this, their rule-based structure provides an inherent decomposition of the approximation into local models, which naturally motivates local parameter estimation \citep{Angelov2004}. 

Unmodeled dynamics, disturbances, approximation errors, and measurement noise may prevent parameter estimates to remain close to the ideal value, potentially leading to parameter drift in the adaptation laws when no robustness is incorporated into them. This is particularly relevant in the case of universal approximators, as an exact representation of the unknown nonlinearities is generally not available, and a nonzero approximation error might remain \citep{Farrel2006}. 

In this case, recursive least squares (RLS) might be more appealing than gradient-based adaptation, since the least squares objective function naturally reduces the parameter estimation sensitivity against these errors \citep{Pan2024,Djaneye-Boundjou2025}. It, however, is not particularly suitable for the long-term operation of adaptive control systems, as over time, it becomes harder and harder to adapt (since it weighs all data it has received over time equally), and a forgetting factor modification is usually employed for it to retain its adaptation capabilities \citep{Goel2020}. This modification introduces a new problem, though. If the regressors used in recursive least squares are not persistently exciting, the covariance may grow without limit, or \emph{winds up}. Even though the prediction error is mathematically proven to still converge to zero \citep[Proposition 6]{Goel2020}, numerical instability caused by an ill-conditioned covariance matrix might lead the estimated parameters to diverge. In that regard, approaches such as Dynamic Regression and Mixing (DREM) \citep{Ortega2020,Ortega2021} and Concurrent Learning \citep{Djaneye-Boundjou2025} allow for softer excitation conditions, whereas approaches like Variable Direction Forgetting \citep{Goel2020} avoid the covariance wind up altogether by only applying forgetting in the direction of new incoming regressors.

Even though these are interesting approaches, all of them generally incur a higher computational cost for the adaptation law, and they still share the same problem as the standard recursive least squares method for large estimation problems: the size of the covariance matrix grows quadratically with the number of estimated parameters. Note that this can be quite a problem when using universal approximators, as the number of estimated parameters can grow rapidly, given a desired approximation error and the number of input dimensions for the approximator. With that in mind, and inspired by the local model networks \citep{Johanssen1995} and evolving Takagi-Sugeno (TS) models \citep{Angelov2004}, this paper proposes a local model estimation using constant consequents (for which any regressor signal is ensured to be persistently exciting for the active rules) and a variation of the recursive least squares with separate covariances for each rule, leading to linear instead of quadratic growth for the covariance.

Another limitation present in the literature is the system structure for which adaptive fuzzy controllers have been developed. A large portion of the results deals either with feedback linearizable or strict feedback systems and relies on either directly canceling the estimated nonlinearity and obtaining a \emph{linear-like} closed loop system \citep{Wang1993,Dass2023} or using a backstepping control law \citep{Xiao2015,Qiu2022,Sun2024,Zhang2024}. This work, instead, relies on the fact that a large class of nonlinear systems can be represented using Takagi-Sugeno fuzzy models, for which Linear Matrix Inequality (LMI) conditions are available for analysis and synthesis problems \citep{Tanaka2004,Nguyen2019}. Inspired by nonlinear consequent models \citep{Dong2009,Coutinho2020,Araujo2021}, we allow for sector-bounded mismatched nonlinearities while still ensuring a uniform bound on the closed loop states. Motivated by output regulation theory \citep{Huang2004}, we also propose, for norm-bounded mismatched nonlinearities, a feedforward term that \emph{approximately} rejects the effect of the estimated nonlinearity on a desired output.

Given these discussions, this articles contributions are as follows:
\begin{itemize}
    \item A novel local model estimation for fuzzy TS models approximating unknown nonlinearities in nonlinear systems is proposed, which ensures that each rule's parameters and covariance are only updated if its membership function is active, along with an asymptotic bound on the local adaptation errors (Theorem \ref{thm:estimador}).
    \item A set of LMI synthesis conditions is presented that ensures the closed loop system is \emph{ultimately uniformly bounded} with regard to the approximation error for the case of matched nonlinearities (Theorem \ref{thm:controle_casado}) and sector-bounded nonlinearities (Theorem \ref{thm:controle_setor}).
    \item A set of LMI synthesis conditions is presented that ensures the closed loop system is \emph{ultimately uniformly bounded} with respect to the approximation error and the adaptation error (independently of the adaptation law used) for the case of norm-bounded nonlinearities (Theorem \ref{thm:controle_norma}) and the design of a feedforward term that, when used in conjunction with said conditions, ensures that \emph{approximately} renders a desired output immune to the estimated nonlinearity (Theorem \ref{thm:Francis_regulator}).
\end{itemize}

\textbf{Notation:} We represent matrix, vector, and scalar variables by uppercase, \textbf{bold text lower case}, and lowercase letters respectively. Though an abuse of notation, some scalar functions will be represented by uppercase variables, either to avoid confusion or due to that being the standard usage in the literature. For symmetric matrices, $M > 0$ ($< 0$) describes that matrix $M$ is positive (negative) definite, and $M \geq 0$ ($\leq 0$) is positive (negative) semidefinite. $M^T$ represents the transpose of a matrix, $M^{-1}$ its inverse, and $M^{-T} = \left(M^T\right)^{-1}=\left(M^{-1}\right)^T$ the inverse of the transpose. When describing a symmetric matrix, the symbol $\ast$ represents block transpose terms that can be inferred from the symmetry. $I$ stands for the identity matrix, whose dimensions can usually be inferred from context and will not be presented, and $\vec{1}_p \in \mathds{R}^p$ represents a $p$-dimensional vector of ones. $\otimes$ stands for the Kronecker product. $\operatorname{diag}(\vec{v})$ is a diagonal matrix whose elements are given by the elements in vector $\vec{v}$. For matrix $A \in \mathds{R}^{n \times m}$, $\operatorname{vec}(A) \in \mathds{R}^{nm}$ stands for the column vector formed by stacking the columns of matrix $A$.

\section{Problem statement}
Consider a nonlinear discrete-time system in the following qLPV (\emph{quasi}-Linear Parameter Varying) form
\begin{align}
    \vec{x}_{k+1} = A(\vec{x}_k)\vec{x_k} + B(\vec{x}_k)\vec{u}_k + G(\vec{x}_k)\vec{\varphi}(\vec{x}_k) \label{eq:sys_qLPV}
\end{align}
with $\vec{x} \in \mathds{R}^n$ the states of the system, $\vec{u} \in \mathds{R}^m$ the control inputs, $A(\vec{x}_k) \in \mathds{R}^{n\times n}, B(\vec{x}_k) \in \mathds{R}^{n \times m}$ and $G(\vec{x}_k) \in \mathds{R}^{n \times s}$ a full-column rank known matrix functions of the state, and $\vec{\varphi}(\vec{x}_k) \in \mathds{R}^s$ an unknown function of the states. In order to derive synthesis conditions for our controllers, we consider a Takagi-Sugeno (TS) fuzzy model describing the same nonlinear system \eqref{eq:sys_qLPV}, of the form
\begin{align}
    \vec{x}_{k+1} = \sum_{i=1}^{r_m} \mu_i(\vec{x}_k) \left(A_i \vec{x}_k + B_i \vec{u}_k + G_i \vec{\varphi}(\vec{x}_k) \right) \label{eq:sys_TS}
\end{align}
with $r_m$ the number of submodels composing our model, $A_i, B_i$ and $G_i$ the local system matrices, and $\mu_i(\vec{x}_k)$ the membership functions that describe the nonlinear system with
\begin{align}
    \mu_i (\vec{x}_k) \geq 0, \; \forall i, \quad \quad \sum_{i=1}^{r_m} \mu_i (\vec{x}_k) = 1.
\end{align}
This model can be seen as a Takagi-Sugeno fuzzy model with nonlinear consequents \cite{Dong2009,Coutinho2020,Araujo2021}, however, in our case, the $\vec{\varphi}(\vec{x}_k)$ nonlinearities are unknown.

In order to deal with these unknown nonlinearities, we will make use of a Takagi-Sugeno fuzzy model with constant consequents. Considering fuzzy inference systems defined by center-average defuzzifier, product inference, and singleton fuzzifier, if the membership functions employed are gaussian membership functions \cite{Wang1993,Wang1997} or with bounded support (such as triangular or trapezoidal membership functions) \cite{Castro1995}, they are known to represent universal approximators (in a compact region). This approximation will take the form
\begin{equation}
    \vec{\varphi}(\vec{x}_k) = \left(\sum_{i=1}^r \left(\dfrac{\bar{h}_i(\vec{x}_k)}{\sum_{j=1}^r \bar{h}_j(\vec{x}_k)}\right) \vec{\theta}_i^T\right) + \vec{\varepsilon}(\vec{x}_k) = \left(\sum_{i=1}^r h_i(\vec{x}_k) \vec{\theta}_i^T\right) + \vec{\varepsilon}(\vec{x}_k) \label{eq:varphi}
\end{equation}
with $r$ the number of rules considered in the TS model, $\vec{\varepsilon}(\vec{x}_k)$ the approximation error, $\vec{\theta}_i^T$ the parameters that minimize the maximum approximation error within the analysis region, $\bar{h}_i(\vec{x}_k)$ are the membership functions (either gaussian or with bounded support), and $h_i(\vec{x}_k)$ the normalized membership functions with the properties
\begin{align}
    h_i (\vec{x}_k) \geq 0, \; \forall i, \quad \quad \sum_{i=1}^{r} h_i (\vec{x}_k) = 1.
\end{align}
If we consider a compact representation
\begin{align}
    \Theta^T = \begin{bmatrix} \vec{\theta}_1^T & \dots & \vec{\theta}_r^T \end{bmatrix}, \quad \vec{h}(\vec{x}_k) = \begin{bmatrix} h_1(\vec{x}_k) \\ \vdots \\ h_r(\vec{x}_k) \end{bmatrix}
\end{align}
it follows that \eqref{eq:varphi} can be rewritten as
\begin{align}
    \vec{\varphi}(\vec{x}_k) = \Theta^T \vec{h}(\vec{x}_k) + \vec{\varepsilon}(\vec{x}_k).
\end{align}

Since the function is unknown, we will try to estimate the $\Theta$ parameter matrix by an $\hat{\Theta}_k$, leading to an estimated nonlinearity of the form
\begin{align}
    \hat{\vec{\varphi}}(\vec{x}_k) = \hat{\Theta}_k^T \vec{h}(\vec{x}_k), && \hat{\Theta}^T = \begin{bmatrix} \hat{\vec{\theta}}_1^T & \dots & \hat{\vec{\theta}}_r^T \end{bmatrix}. \label{eq:adapt_nonlinearity}
\end{align}
By adding and subtracting $G(\vec{x}_k) \hat{\vec{\varphi}}(\vec{x}_k)$ to the open-loop system dynamics in \eqref{eq:sys_qLPV}, we end up with
\begin{align}
    \vec{x}_{k+1} &= A(\vec{x}_k)\vec{x_k} + B(\vec{x}_k)\vec{u}_k + G(\vec{x}_k) \hat{\vec{\varphi}}(\vec{x}_k) \\&\quad+ G(\vec{x}_k) \left(\tilde{\Theta}_k^T\vec{h}(\vec{x}_k) + \vec{\varepsilon}(\vec{x}_k)\right) \label{eq:ol_estimation}
\end{align}
with $\tilde{\Theta}_k^T = \Theta^T - \hat{\Theta}_k^T$, and $\tilde{\Theta}^T = \begin{bmatrix} \tilde{\vec{\theta}}_1^T & \dots & \tilde{\vec{\theta}}_r^T \end{bmatrix}$.

\section{Local Model Estimation}
From the normalized membership functions of the TS approximator being used, we can rewrite the open-loop \eqref{eq:ol_estimation} as
\begin{align}
    \vec{x}_{k+1} = \sum_{i=1}^r h_i(\vec{x}_k) &\left(A(\vec{x}_k)\vec{x}_k + B(\vec{x}_k)\vec{u}_k + G(\vec{x}_k) \hat{\vec{\theta}}_{ik}^T \right. \\[-0.2cm] &\left.+ G(\vec{x}_k) \left(\tilde{\vec{\theta}}_{ik}^T + \vec{\varepsilon}(\vec{x}_k)\right)\right).\label{eq:open_loop}
\end{align}

Defining the local estimation error as 
\begin{align}
    \vec{e}^{(i)}_{k+1} &= \vec{x}_{k+1} - A(\vec{x}_k)\vec{x}_k - B(\vec{x}_k)\vec{u}_k - G(\vec{x}_k) \hat{\vec{\theta}}_{ik}^T, \\
\end{align}
with $(i)$ representing the corresponding rule error, we have that
\begin{align}
    \sum_{i=1}^r h_i(\vec{x}_k) \vec{e}^{(i)}_{k+1}  = \sum_{i=1}^r h_i(\vec{x}_k) G(\vec{x}_k) \left(\tilde{\vec{\theta}}_{ik}^T + \vec{\varepsilon}(\vec{x}_k)\right). \label{eq:erro_global}
\end{align}
Even though equation \eqref{eq:erro_global} describes the global approximation error of the TS approximator, its use could lead to considerably large covariance matrices (if directly employed on a recursive least-squares estimator). Instead, we focus our attention on the local errors weighted by their respective membership functions separately
\begin{align}
     h_i(\vec{x}_k) \vec{e}^{(i)}_{k+1}  &= h_i(\vec{x}_k) G(\vec{x}_k) \left(\tilde{\vec{\theta}}_{ik}^T + \vec{\varepsilon}(\vec{x}_k) + \sum_{j = 1}^r h_j(\vec{x}_k)\left(\theta_j - \theta_i\right)\right), \\
     h_i(\vec{x}_k) \vec{e}^{(i)}_{k+1}  &= h_i(\vec{x}_k) G(\vec{x}_k) \left(\tilde{\vec{\theta}}_{ik}^T + \vec{\varepsilon}_i(\vec{x}_k)\right), \label{eq:modelo_local}
\end{align}
with $\vec{\varepsilon}_i(\vec{x}_k) = \vec{\varepsilon}(\vec{x}_k) + \sum_{j = 1}^r h_j(\vec{x}_k)\left(\theta_j - \theta_i\right)$ the local approximation error, leading to a local analysis of the approximation error. 

This formulation is consistent with the bounded support or gaussian membership functions employed, as they naturally restrict the validity of each local model to a subset of the state space. This locality principle is not unique to our proposition, but is also a fundamental property employed in local model networks and evolving Takagi-Sugeno identification frameworks, in which the parameter estimation is inherently local and no global error minimization is sought \cite{Johanssen1995,Angelov2004}.

To shorten the notation a bit, from this point on, we will consider the following short notation
\begin{align}
    h_{ik} = h_i(\vec{x}_k), && \vec{\varepsilon}_k = \vec{\varepsilon}(\vec{x}_k), && \vec{\varepsilon}_{ik} = \vec{\varepsilon}_i(\vec{x}_k).
\end{align}

Having presented these steps, we will consider the following \emph{membership-function weighted} local \emph{recursive least squares}, with forgetting factor $\alpha \in (0,1)$, for the adaptation law (which is considerably simplified using the local models and constant consequents in the TS approximator)
\begin{align}
    p_{i(k+1)} &= \dfrac{p_{ik}}{ 1-h_{ik}\alpha + h_{ik} p_{ik}} \label{eq:cov}\\
    \hat{\vec{\theta}}^T_{i(k+1)} &= \hat{\vec{\theta}}^T_{ik} + p_{i(k+1)} (G^T(\vec{x}_k) G(\vec{x}_k))^{-1} G^T(\vec{x}_k) h_{ik} \vec{e}^{(i)}_{k+1}   \label{eq:mqr_local}
\end{align}
with $p_{ik}$ the estimated covariance.

The covariance equation presented is such that its inverse has the following dynamics
\begin{align}
    p_{i(k+1)}^{-1} &= (1 - h_{ik}\alpha)p_{ik}^{-1} + h_{ik}. \label{eq:cov_inv}
\end{align}
Note that the formulation presented is such that, if $h_{ik} = 0$, then $p_{i(k+1)} = p_{ik}$ and $\hat{\vec{\theta}}_{i(k+1)} = \hat{\vec{\theta}}_{ik}$. Therefore, \textbf{as long as rule $i$ is not active, its parameters and covariance cease to be adapted}, and we guarantee that the models are only adapted \textbf{locally}.

In addition to this, note that, if we make the following change of coordinates $w_{ik} = p_{ik}^{-1} - \frac{1}{\alpha}$, we end up with the dynamics
\begin{align}
    w_{i(k+1)} &= p_{i(k+1)}^{-1} - \frac{1}{\alpha} \\
    w_{i(k+1)} &= (1 - h_{ik}\alpha)p_{ik}^{-1} + h_{ik} - \frac{1}{\alpha} \\
    w_{i(k+1)} &= (1 - h_{ik}\alpha)\left(w_{ik} + \frac{1}{\alpha}\right) + h_{ik} - \frac{1}{\alpha} \\
    w_{i(k+1)} &= (1 - h_{ik}\alpha)w_{ik}
\end{align}
which ensures that, whenever $h_{ik} \neq 0$, $w_{ik}$ converges monotonically to zero, which in turn implies that, whenever $h_{ik} \neq 0$, $p_{ik}^{-1}$ converges monotonically to $1/\alpha$, and $p_{ik}$ converges monotonically to $\alpha$.

From the adaptation law presented in \eqref{eq:mqr_local} and the local error behavior in \eqref{eq:modelo_local}, we get that
\begin{align}
    \tilde{\vec{\theta}}^T_{i(k+1)} &= (1 - p_{i(k+1)} h_{ik})\tilde{\vec{\theta}}^T_{ik} - p_{i(k+1)} h_{ik} \vec{\varepsilon}_{ik}
\end{align}

Defining, for each local model, a Lyapunov function of the form
\begin{align}
    V_{\theta_{ik}} = \frac{1}{p_{ik}} \tilde{\vec{\theta}}_{ik} \tilde{\vec{\theta}}_{ik}^T, \label{eq:local_lyap_est}
\end{align}
it follows that their one-step difference, $\Delta V_{\theta_i} = V_{\theta_{i(k+1)}} - V_{\theta_{ik}}$, can be written as
\begin{align}
    \Delta V_{\theta_i} = &\frac{(1 - p_{i(k+1)}h_{ik})^2}{p_{i(k+1)}} \tilde{\vec{\theta}}_{ik} \tilde{\vec{\theta}}_{ik}^T - \frac{2(1 - p_{i(k+1)}h_{ik}) p_{i(k+1)}h_{ik}}{p_{i(k+1)}} \tilde{\vec{\theta}}_{ik} \vec{\varepsilon}_{ik} \\
    &+ p_{i(k+1)}h_{ik}^2 \vec{\varepsilon}_{ik}^T \vec{\varepsilon}_{ik} 
    - \frac{1}{p_{ik}} \tilde{\vec{\theta}}_{ik} \tilde{\vec{\theta}}_{ik}^T
\end{align}

Considering that, given $\beta > 0$,
\begin{align}
    - \frac{2(1 - p_{i(k+1)}h_{ik})p_{i(k+1)}h_{ik}}{p_{i(k+1)}} \tilde{\vec{\theta}}_{ik} \vec{\varepsilon}_{ik}\!\leq\!\beta \frac{(1 - p_{i(k+1)}h_{ik})^2}{p_{i(k+1)}} \tilde{\vec{\theta}}_{ik}\tilde{\vec{\theta}}_{ik}^T\!+\!\frac{1}{\beta} p_{i(k+1)}h^2_{ik} \vec{\varepsilon}_{ik}^T \vec{\varepsilon}_{ik}
\end{align}
we can write that
\begin{align}
    \Delta V_{\theta_i} \leq \frac{(1 + \beta)(1 - p_{i(k+1)}h_{ik})^2}{p_{i(k+1)}} \tilde{\vec{\theta}}_{ik} \tilde{\vec{\theta}}_{ik}^T - \frac{1}{p_{ik}} \tilde{\vec{\theta}}_{ik} \tilde{\vec{\theta}}_{ik}^T + \frac{\beta + 1}{\beta} p_{i(k+1)}h_{ik}^2 \vec{\varepsilon}_{ik}^T \vec{\varepsilon}_{ik} 
\end{align}
which, after some algebraic manipulations and making use of the covariance and inverse covariance dynamics \eqref{eq:cov} and \eqref{eq:cov_inv}, leads to
\begin{align}
    \Delta V_{\theta_i} \leq& \dfrac{-(1-\alpha h_{ik}) - h_{ik}p_{ik} + \left(1 + \beta\right)(1 - \alpha h_{ik})^2 }{p_{ik}((1 - \alpha h_{ik}) + h_{ik}p_{ik})} \tilde{\vec{\theta}}_{ik} \tilde{\vec{\theta}}_{ik}^T\\
    &+ \dfrac{(\beta+1) p_{ik}  h_{ik}^2}{\beta ((1 - \alpha h_{ik}) + h_{ik}p_{ik}) } \vec{\varepsilon}_{ik}^T \vec{\varepsilon}_{ik} 
\end{align}

If we choose, for $\gamma \in (0,1)$,
\begin{align}
    \beta &=  \dfrac{(1 -\gamma)\alpha h_{ik}}{1 - (1 - \gamma)\alpha h_{ik}} \label{eq:def_beta}
\end{align}
then
\begin{align}
    &(1 + \beta) \left(\dfrac{-(1-\alpha h_{ik})}{1 + \beta} + (1 - \alpha h_{ik})^2 - \dfrac{hh_{ik}p_{ik}}{1 + \beta}\right) \\ &\quad= -\dfrac{\gamma \alpha h_{ik}}{1 - (1 - \gamma)\alpha h_{ik}} (1 - \alpha h_{ik} + h_{ik}p_{ik}) - \dfrac{h_{ik}p_{ik}}{1 - (1 - \gamma)\alpha h_{ik}} (1 - \alpha h_{ik}), \\
\end{align}
and 
\begin{align}
    &\dfrac{(\beta+1) p_{ik} h_{ik}^2}{\beta ((1 - \alpha h_{ik}) + h_{ik}p_{ik}) } = \dfrac{p_{ik} h_{ik}}{(1 - \gamma)\alpha((1 - \alpha h_{ik}) + h_{ik}p_{ik})},
\end{align}
such that
\begin{align}
    \Delta V_{\theta_i} \leq& -\dfrac{\gamma \alpha h_{ik}}{p_{ik}\left(1 - (1 - \gamma)\alpha h_{ik}\right)} \tilde{\vec{\theta}}_{ik} \tilde{\vec{\theta}}_{ik}^T   + \dfrac{p_{ik} h_{ik}}{(1 - \gamma)\alpha((1 - \alpha h_{ik}) + h_{ik}p_{ik})} \vec{\varepsilon}_{ik}^T \vec{\varepsilon}_{ik}
    \\&- \dfrac{h_{ik}(1 - \alpha h_{ik})}{\left(1 - \alpha h_{ik} + h_{ik}p_{ik}\right)\left(1 - (1 - \gamma)\alpha h_{ik}\right)} \tilde{\vec{\theta}}_{ik} \tilde{\vec{\theta}}_{ik}^T,\\
    \Delta V_{\theta_i} \leq& h_{ik} \left(- \dfrac{\gamma \alpha}{1 - (1 - \gamma)\alpha h_{ik}} V_{\theta_{ik}}  + \dfrac{p_{ik}}{(1 - \gamma)\alpha((1 - \alpha h_{ik}) + h_{ik}p_{ik})} \vec{\varepsilon}_{ik}^T \vec{\varepsilon}_{ik} \right. \\ 
    &\left. - \dfrac{1 - \alpha h_{ik}}{\left(1 - \alpha h_{ik} + h_{ik}p_{ik}\right)\left(1 - (1 - \gamma)\alpha h_{ik}\right)} \tilde{\vec{\theta}}_{ik}\tilde{\vec{\theta}}_{ik}^T \right).\label{eq:dvt_i}
\end{align}

Since 
\begin{equation}
    \dfrac{1}{1 - (1 - \gamma)\alpha h_{ik}} > 1,
\end{equation}
and we have shown that, whenever $h_{ik} \neq 0$, $p_{ik}$ converges monotonically to $\alpha$, it follows that $p_{ik} \leq \bar{p}$, for $p_{i0} = \bar{p}$ with $\bar{p} \geq \alpha$. Therefore,  $h_{ik}p_{ik} \geq \alpha h_{ik}$ and
\begin{align}
    \Delta V_{\theta_i} \leq& h_{ik} \left(- \gamma \alpha V_{\theta_{ik}} + \dfrac{\bar{p}}{(1 - \gamma)\alpha} \vec{\varepsilon}_{ik}^T \vec{\varepsilon}_{ik} \right).
\end{align}
Considering that $\vec{\varepsilon}_{ik}^T \vec{\varepsilon}_{ik} \leq \delta$, we have that
\begin{align}
    V_{\theta_{i(k+1)}} \leq (1 - h_{ik} \gamma \alpha) V_{\theta_{ik}} + \dfrac{h_{ik} \bar{p}}{(1 - \gamma)\alpha} \delta.
\end{align}
Similarly to what we did to prove the convergence of $p_{ik}$, consider the change of variables 
\begin{align}
    W_{\theta_{ik}} = V_{\theta_{ik}} - \dfrac{\bar{p} \delta}{\gamma (1 - \gamma)\alpha^2}
\end{align}
which implies that
\begin{align}
    W_{\theta_{i(k+1)}} \leq& (1 - h_{ik} \gamma \alpha) V_{\theta_{ik}} + \dfrac{h_{ik} \bar{p}}{(1 - \gamma)\alpha} \delta - \dfrac{\bar{p} \delta}{\gamma (1 - \gamma)\alpha^2} \\
    W_{\theta_{i(k+1)}} \leq& (1 - h_{ik} \gamma \alpha) \left(W_{\theta_{ik}} + \dfrac{\bar{p} \delta}{\gamma (1 - \gamma)\alpha^2}\right) + \dfrac{h_{ik} \bar{p}}{(1 - \gamma)\alpha} \delta - \dfrac{\bar{p} \delta}{\gamma (1 - \gamma)\alpha^2} \\
    W_{\theta_{i(k+1)}} \leq& (1 - h_{ik} \gamma \alpha) W_{\theta_{ik}}
\end{align}
and ensures that $W_{\theta_{ik}}$ is bounded and, whenever $h_{ik} \neq 0$, its upper bound converges monotonically such that, ultimately, $W_{\theta_{ik}} \leq 0$. It follows that $V_{\theta_{ik}}$ is bounded and, whenever $h_{ik} \neq 0$, it monotonically converges such that 
\begin{align}
    0 \leq V_{\theta_{ik}} \leq \dfrac{\bar{p} \delta}{\gamma (1 - \gamma)\alpha^2}. \label{eq:v_bound}
\end{align}

\begin{thm} \label{thm:estimador}
    The \emph{membership-function weighted} local \emph{recursive} \emph{least} \emph{squares} with a forgetting factor $\alpha$, presented in equations \eqref{eq:cov} and \eqref{eq:mqr_local}, ensures that the adaptation error is bounded and, if $h_{ik} \neq 0$, converges so that
    \begin{align}
        \|\tilde{\vec{\theta}}_{ik}\|_2^2 \leq 4\delta.
    \end{align}
\end{thm}
\begin{proof}
    The bound follows from choosing $\gamma = 0.5$, the maximum of $\gamma (1 - \gamma)$, and considering that a tighter bound can be found for \eqref{eq:v_bound} by noting that, for $h_{ik} \neq 0$, $p_{ik}$ converges to $\alpha$.
\end{proof}

\section{LMI Controller Synthesis}
In order to develop our synthesis conditions for the control law, we will make use of the following Lemma.
\begin{lem} \label{lem:LyapX}
    If there exists a control law for the system in \eqref{eq:open_loop}, such that, for given positive scalars $\lambda$, $\eta$, and $\nu$, there exists a Lyapunov function, $V_{x_k}$, satisfying
    \begin{align}
        V_{x_k} &\geq \lambda \vec{x}_k^T \vec{x}_k \\
        V_{x_{k+1}} - V_{x_k} = \Delta V_x &\leq - \eta V_{x_k} + \nu \left(\sum_{i=1}^r h_{ik}^2 \tilde{\vec{\theta}}_{ik} \tilde{\vec{\theta}}_{ik}^T \right) + \nu \vec{\varepsilon}_k^T \vec{\varepsilon}_k
    \end{align}
    then the closed loop employing said control law with the adaptation law in \eqref{eq:cov} and \eqref{eq:mqr_local} ensures that the system states, $\vec{x}_k$ will be \emph{ultimately uniformly bounded} with bound
    \begin{align}
        \|x_k\|_2^2 \leq \nu \delta \left(1 + \dfrac{4 \bar{p}^2}{\alpha^2} \right) \left(\dfrac{1 + \eta}{\eta \lambda} \right) \label{eq:UUB_bound}
    \end{align}
    with $\bar{p} \geq \alpha$ the maximum value chosen for the initial condition of the covariances, $\vec{\varepsilon}_{k}^T \vec{\varepsilon}_{k} \leq \delta$, and $\vec{\varepsilon}_{ik}^T \vec{\varepsilon}_{ik} \leq \delta$.
\end{lem}
\begin{proof}
    Consider a Lyapunov function for the whole closed loop system (states and estimation error) given by
    \begin{align}
        V_k = V_{x_k} + \frac{\bar{p} \nu}{\gamma \alpha} \sum_{i = 1}^r V_{\theta_{ik}}.
    \end{align}
    It follows that its one-step difference, $\Delta V = V_{k+1} - V_k$, will be upper-bounded by
    \begin{align}
        \Delta V \leq &- \eta V_{x_k} + \nu \left(\sum_{i=1}^r h_{ik}^2 \tilde{\vec{\theta}}_{ik} \tilde{\vec{\theta}}_{ik}^T \right) + \nu \vec{\varepsilon}_k^T \vec{\varepsilon}_k + \dfrac{\nu\bar{p}^2}{(1 - \gamma)\gamma\alpha^2} \sum_{i=1}^r h_{ik} \vec{\varepsilon}_{ik}^T \vec{\varepsilon}_{ik} \\
        &- \nu \sum_{i=1}^r \left( h_{ik}  \tilde{\vec{\theta}}_{ik} \tilde{\vec{\theta}}_{ik}^T \right), \\
        \Delta V \leq &- \eta V_{x_k} + \nu \left(\sum_{i=1}^r h_{ik}^2 \tilde{\vec{\theta}}_{ik} \tilde{\vec{\theta}}_{ik}^T \right) - \nu \sum_{i=1}^r \left( h_{ik}  \tilde{\vec{\theta}}_{ik} \tilde{\vec{\theta}}_{ik}^T \right) + \bar{\nu} \delta,
    \end{align}
    with $\bar{\nu} = \nu \left(1 + \dfrac{\bar{p}^2}{(1 - \gamma)\gamma\alpha^2} \right)$.
    
    Since $h_{ik}^2 \tilde{\vec{\theta}}_{ik} \tilde{\vec{\theta}}_{ik}^T \leq h_{ik} \tilde{\vec{\theta}}_{ik} \tilde{\vec{\theta}}_{ik}^T$, we have that
    \begin{align}
         \Delta V \leq - \eta V_{x_k} + \bar{\nu} \delta.
    \end{align}

    This is guaranteed to be less than or equal to zero as long as $V_{x_k} \geq \frac{\bar{\nu}\delta}{\eta}$, leading the states to the region $V_{x_k} \leq \frac{\bar{\nu}\delta}{\eta}$. From within this region, the maximum change that can happen to the Lyapunov function is $\bar{\nu} \delta$, which is outside of the region and will lead to a decrease of the Lyapunov function once again. Hence, the system is bounded such that
    \begin{align}
        \lambda \vec{x}_k^T \vec{x}_k \leq V_{x_k} \leq \frac{\eta + 1}{\eta}\bar{\nu}\delta,
    \end{align}
    and the bounded region in \eqref{eq:UUB_bound} is found by considering $\gamma = 0.5$.
\end{proof}

\begin{rem}
    A tighter bound can be found by considering that $p_{ik}$ converges to $\alpha$ (or by making use of the bound in Theorem \ref{thm:estimador})
    \begin{align}
        \|x_k\|_2^2 \leq 5 \nu \delta \left(\dfrac{1 + \eta}{\eta \lambda} \right).
    \end{align}
\end{rem}

With Lemma \ref{lem:LyapX}, we now have a tool to develop adaptive control laws for the system in different cases, which will be discussed in the following.

\subsection{Matched Nonlinearity}
We first focus our attention on the simpler case of an unknown matched nonlinearity, which will be summarized in the following assumption.
\begin{asm}[Matched Nonlinearity] \label{asm:matched}
    Matrix $G(\vec{x}_k)$, in \eqref{eq:sys_qLPV}, can be decomposed as
    \begin{align}
        G(\vec{x}_k) = B(\vec{x}_k) F(\vec{x}_k). \label{eq:asm_match}
    \end{align}
\end{asm}

If Assumption \ref{asm:matched} is met, we can employ an adaptive control law of the form
\begin{align}
    \vec{u}_k &= K(\vec{x}_k) \vec{x}_k - F(\vec{x}_k) \hat{\vec{\varphi}}(\vec{x}_k) \label{eq:control_law}
\end{align}
with $K(\vec{x}_k)$ the state feedback gains, $F(\vec{x}_k)$ the estimated nonlinearity feedback control gains in equation \eqref{eq:open_loop}, such that the closed loop dynamics can be written as
\begin{align}
    \vec{x}_{k+1} &= \sum_{i=1}^r h_i(\vec{x}_k) \left((A(\vec{x}_k) + B(\vec{x}_k) K(\vec{x}_k))\vec{x}_k + G(\vec{x}_k) \left(\tilde{\vec{\theta}}_{ik}^T + \vec{\varepsilon}(\vec{x}_k)\right)\right) \\
    \vec{x}_{k+1} &= (A(\vec{x}_k) + B(\vec{x}_k) K(\vec{x}_k))\vec{x}_k + G(\vec{x}_k) \left(\vec{\varepsilon}(\vec{x}_k) + \sum_{i=1}^r h_i(\vec{x}_k) \tilde{\vec{\theta}}_{ik}^T\right) \\
    \vec{x}_{k+1} &= (A(\vec{x}_k) + B(\vec{x}_k) K(\vec{x}_k))\vec{x}_k + \left(\vec{1}_{r+1}^T \otimes G(\vec{x}_k) \right) \begin{bmatrix}\vec{\varepsilon}(\vec{x}_k) \\ h_{1k} \tilde{\vec{\theta}}_{1k}^T \\ \vdots \\ h_{rk} \tilde{\vec{\theta}}_{rk}^T \end{bmatrix} \\
    \vec{x}_{k+1} &= (A(\vec{x}_k) + B(\vec{x}_k) K(\vec{x}_k))\vec{x}_k + \left(\vec{1}_{r+1}^T \otimes G(\vec{x}_k) \right) \vec{\zeta}_k
\end{align}
with $\vec{\zeta}_k^T = \begin{bmatrix}\vec{\varepsilon}^T(\vec{x}_k) & h_{1k} \tilde{\vec{\theta}}_{1k} & \dots & h_{rk} \tilde{\vec{\theta}}_{rk} \end{bmatrix}$.

If we consider a Parallel Distributed Compensation (PDC) control law, we will assume that the feedback control gains $K(\vec{x}_k)$ are such that 
\begin{align}
    K(\vec{x}_k) = \sum_{i=1}^{r_m} \mu_i(\vec{x}_k) K_i, \label{eq:K_PDC}
\end{align}
so that the closed loop system can be represented as
\begin{align}
        \vec{x}_{k+1} = \sum_{i=1}^{r_m} \sum_{j=1}^{r_m}\mu_i(\vec{x}_k) \mu_j(\vec{x}_k) \left((A_i + B_i K_j)\vec{x}_k + \left(\vec{1}_{r+1}^T \otimes G_i \right) \vec{\zeta}_k\right).
\end{align}

By considering a fuzzy Lyapunov function for $V_{x_k}$ of the form
\begin{align}
    V_{x_k} = \vec{x}_k^T \left(\sum_{i=1}^{r_m} \mu_i(\vec{x}_k) P_i\right)^{-1} \vec{x}_k = \vec{x}_k^T P^{-1}(\vec{x}_k) \vec{x}_k
\end{align}
we need to impose the conditions from Lemma \ref{lem:LyapX} in order to derive our LMI conditions. As such, we first impose that
\begin{align}
    \vec{x}_k^T P^{-1}(\vec{x}_k) \vec{x}_k \geq \lambda \vec{x}_k^T \vec{x_k} 
\end{align}
which leads to the matrix inequalities
\begin{align}
    \left(\sum_{i=1}^{r_m} \mu_i(\vec{x}_k) P_i\right)^{-1} &\geq \lambda I \\
    \sum_{i=1}^{r_m} \mu_i(\vec{x}_k) P_i &\leq \frac{1}{\lambda} I
\end{align}
for which a sufficient condition is that
\begin{align}
    0 < P_i \leq \frac{1}{\lambda} I, \quad \forall i
\end{align}

For the second set of conditions from the Lemma, it suffices that
\begin{align}
    \vec{x}_{k+1}^T P^{-1}(\vec{x}_{k+1}) \vec{x}_{k+1}   - \vec{x}_k^T P^{-1}(\vec{x}_k) \vec{x}_k \leq -\eta \vec{x}_k^T P^{-1}(\vec{x}_k) \vec{x}_k + \nu \vec{\zeta}_k^T \vec{\zeta}_k 
    \end{align}
    \begin{align}
    \begin{bmatrix} \vec{x}_k \\ \zeta_k \end{bmatrix}^T \left(\begin{bmatrix} A_{c}^T(\vec{x}_k) \\ \Xi_1(\vec{x}_k) \end{bmatrix} P^{-1}(\vec{x}_{k+1}) \begin{bmatrix} A_c^T(\vec{x}_k) \\ \Xi_1(\vec{x}_k) \end{bmatrix}^T \right.\\
    \left.
    +\begin{bmatrix} -(1-\eta)P^{-1}(\vec{x}_k) & 0 \\ \ast & -\nu I\end{bmatrix}
    \right)\begin{bmatrix} \vec{x}_k \\ \zeta_k \end{bmatrix} \leq 0
\end{align}
with $A_c(\vec{x}_k) = A(\vec{x}_k) + B(\vec{x}_k) K(\vec{x}_k)$ and $\Xi_1(\vec{x}_k)=\vec{1}_{r+1} \otimes G^T(\vec{x}_k)$. This quadratic form, in turn, leads to the matrix inequalities
\begin{align}
    \begin{bmatrix} -(1-\eta)P^{-1}(\vec{x}_k)  & 0 \\ \ast & -\nu I\end{bmatrix} + \begin{bmatrix} A_c^T(\vec{x}_k)\\ \Xi_1(\vec{x}_k) \end{bmatrix} P^{-1}(\vec{x}_{k+1}) \begin{bmatrix} A_c^T(\vec{x}_k) \\  \Xi_1(\vec{x}_k) \end{bmatrix}^T \leq 0,
\end{align}
which from Schur's complement is equivalent to
\begin{align}
    \begin{bmatrix} -(1-\eta)P^{-1}(\vec{x}_k)  & 0  & A_c^T(\vec{x}_k) \\ \ast & -\nu I & \Xi_1(\vec{x}_k) \\ \ast & \ast & -P(\vec{x}_{k+1}) \end{bmatrix} \leq 0.
\end{align}

Multiplying it on the left and on the right by $\begin{bmatrix} I & 0 & 0 \\ 0 & 0 & I \\ 0 & I & 0 \end{bmatrix}$, leads to
\begin{align}
    \begin{bmatrix} -(1-\eta)P^{-1}(\vec{x}_k)  & A_c^T(\vec{x}_k) & 0  \\ \ast & -P(\vec{x}_{k+1}) & \Xi_1^T(\vec{x}_k) \\ \ast & \ast & -\nu I \end{bmatrix} \leq 0
\end{align}
which from Schur's complement, and using the fact that $(X \otimes W)(Y \otimes Z) = (XY \otimes WZ)$, is equivalent to
\begin{align}
    &\begin{bmatrix} -(1-\eta)P^{-1}(\vec{x}_k)  & A_c^T(\vec{x}_k) \\ \ast & -P(\vec{x}_{k+1}) \end{bmatrix} + \dfrac{1}{\nu}\begin{bmatrix} 0 \\ \Xi_1^T(\vec{x}_k) \end{bmatrix} \begin{bmatrix} 0 &  \Xi_1(\vec{x}_k)\end{bmatrix} \leq 0 
\end{align}
\begin{align}
    \begin{bmatrix} -(1-\eta)P^{-1}(\vec{x}_k)  & A_c^T(\vec{x}_k) \\ \ast & -P(\vec{x}_{k+1}) \end{bmatrix} + \dfrac{r+1}{\nu}\begin{bmatrix} 0 \\ G(\vec{x}_k) \end{bmatrix} \begin{bmatrix} 0 &  G^T(\vec{x}_k)\end{bmatrix} \leq 0,
\end{align}
and from another Schur's complement is equivalent to
\begin{align}
    &\begin{bmatrix} -(1-\eta)P^{-1}(\vec{x}_k)  & A^T(\vec{x}_k) + K^T(\vec{x}_k)B^T(\vec{x}_k) & 0  \\ \ast & -P(\vec{x}_{k+1}) & G(\vec{x}_k) \\ \ast & \ast & -\frac{\nu}{r+1} I \end{bmatrix} \leq 0.
\end{align}

Multiplying it on the left by $\diag{M^T, I, I}$ and on the right by its transpose, by defining $\bar{K}(\vec{x}_k) = K(\vec{x}_k) M$, and making use of the fact that $-M^T P^{-1}(\vec{x}_k) M $ $\leq -M^T -M + P(\vec{x}_k)$, we end up with the sufficient condition
\begin{align}
    \begin{bmatrix} -(1-\eta)\left(M^T + M - P(\vec{x}_k)\right)  & \ast & \ast  \\ A(\vec{x}_k)M+B(\vec{x}_k)\bar{K}(\vec{x}_k) & -P(\vec{x}_{k+1}) & \ast \\ 0 & G^T(\vec{x}_k) & -\frac{\nu}{r+1} I \end{bmatrix} \leq 0.
\end{align}

Given the steps taken in this subsection, Theorem \ref{thm:controle_casado} summarizes the synthesis results for the case of matched uncertainties.
\begin{thm} \label{thm:controle_casado}
    Consider the nonlinear discrete-time system in \eqref{eq:sys_qLPV}, satisfying Assumption \ref{asm:matched} and admitting the TS fuzzy model description in \eqref{eq:sys_TS}. If there exist symmetric matrices $P_i \in \mathds{R}^{n \times n}, \; i \in [1, \dots, r_m]$, matrices $M \in \mathds{R}^{n \times n}$ and $\bar{K}_i \in \mathds{R}^{m \times n}, \; i \in [1, \dots, r_m]$ and a scalar $\nu$ such that
    \begin{align}
        0 < P_{i} &\leq \dfrac{1}{\lambda} I, \quad \forall i \\
        Q_{ii\ell} &\leq 0, \quad \forall i, \ell \\
        Q_{ij\ell} + Q_{ji\ell} & \leq 0, \quad \forall \ell, i, j> i
    \end{align}
    with
    \begin{align}
        Q_{ij\ell} = \begin{bmatrix} -(1-\eta)\left(M^T + M - P_i\right)  & M^T A_i^T + \bar{K}_j^T B_i^T & 0  \\ \ast & -P_\ell & G_i \\ \ast & \ast & -\frac{\nu}{r+1} I \end{bmatrix}
    \end{align}
    then the adaptive control law \eqref{eq:control_law} (with $K(\vec{x}_k)$ given by \eqref{eq:K_PDC}, $K_i = \bar{K}_i M^{-1}$, and $F(\vec{x}_k)$ satisfying \eqref{eq:asm_match}) together with the adaptation law in \eqref{eq:cov} and \eqref{eq:mqr_local} ensures that the states $\vec{x}_k$ of the closed loop system are \emph{ultimately uniformly bounded} with bound \eqref{eq:UUB_bound}.
\end{thm}

\subsection{Sector-bounded Nonlinearities}
At this point, we no longer consider the case of matched nonlinearities but focus our attention on sector-bounded uncertain nonlinearities. We consider that each of the unknown terms in the unknown function $\vec{\varphi}(\vec{x})$ admits a known sector bound, as summarized in the following assumption.
\begin{asm}[Sector-bounded Nonlinearity] \label{asm:sector}
    The terms composing the unknown function $\vec{\varphi}(\vec{x}_k)$ admit known sector bounds, such that
    \begin{align}
        \vec{\phi}_{1j}^T \vec{x}_k \leq \varphi_j(\vec{x}_k) \leq \vec{\phi}_{2j}^T \vec{x}_k,
    \end{align}
    with $\varphi_j(\vec{x}_k)$ the $j$-element of $\vec{\varphi}(\vec{x}_k)$, and $\vec{\phi}_{1j}^T$ and $\vec{\phi}_{2j}^T$ vectors relating to its lower and upper bounds.
\end{asm}

Since the nonlinearities are being approximated by our adaptation law, we will assume that these same bounds need to be respected by our approximator $\hat{\vec{\varphi}}(\vec{x}_k) = \hat{\Theta}_k^T \vec{h}(\vec{x}_k)$. With that in mind, it makes sense to introduce constraints on our adaptation so that these bounds will be respected. Considering our local approximation approach, these bounds can be represented as
\begin{align}
    h_{ik} \vec{\phi}_{1j}^T \vec{x}_k \leq h_{ik} \hat{\theta}_{jik} \leq h_{ik} \vec{\phi}_{2j}^T \vec{x}_k \label{eq:bounds_memb_sector}
\end{align}
with $\hat{\theta}_{jik}$ the element on row $j$ and column $i$ from our parameter matrix $\hat{\Theta}_k$. Equation \eqref{eq:bounds_memb_sector} indicates that the bounds are only considered whenever the membership function corresponding to a certain parameter is active ($h_{ik} \neq 0$); otherwise, they are disregarded. In order to impose these constraints, the following projection modification is made to the adaptation law in \eqref{eq:mqr_local}
\begin{align}
     \vec{\xi}^T_{i(k+1)} &= \hat{\vec{\theta}}^T_{ik} + p_{i(k+1)} (G^T(\vec{x}_k) G(\vec{x}_k))^{-1} G^T(\vec{x}_k) h_{ik} \vec{e}^{(i)}_{k+1}  \label{eq:xi_proj}\\
     \hat{\theta}_{ji(k+1)} &= \left\{\begin{array}{ll} \xi_{ji(k+1)}, & \textrm{if } h_{ik} = 0 \textrm{ or } \vec{\phi}_{1j}^T \vec{x}_k \leq \xi_{ji(k+1)} \leq \vec{\phi}_{2j}^T \vec{x}_k, \\[0.12cm] 
     \vec{\phi}_{1j}^T \vec{x}_k, & \textrm{if } \xi_{ji(k+1)} < \vec{\phi}_{1j}^T \vec{x}_k, \\[0.12cm]
     \vec{\phi}_{2j}^T \vec{x}_k, & \textrm{if } \xi_{ji(k+1)} > \vec{\phi}_{2j}^T \vec{x}_k.\end{array} \right. \label{eq:theta_proj}
\end{align}

Since we consider that the original $\Theta$ parameters being estimated are within this region, the use of this projection always leads to a decrease in the adaptation error. In addition, since our covariances are always described by scalars, this decrease directly corresponds to a decrease in the local adaptation Lyapunov function \eqref{eq:local_lyap_est}, and the terms introduced by the projection can be ignored for stability/convergence analysis, implying that Theorem \ref{thm:estimador} remains valid.

As proposed in \cite{Dong2009} and adapted in \cite{Coutinho2020,Araujo2021}, if Assumption \ref{asm:sector} is satisfied, it follows that
\begin{align}
    \left(\vec{\phi}_{1j}^T \vec{x}_k - \hat{\varphi}_j(\vec{x}_k)\right)\left(\vec{\phi}_{2j}^T \vec{x}_k - \hat{\varphi}_j(\vec{x}_k)\right) \leq 0.
\end{align}

By defining $\vec{\phi}_{j}^T = \vec{\phi}_{2j}^T - \vec{\phi}_{1j}^T$ and $\bar{\varphi}_j(\vec{x}_k) = \hat{\varphi}_j(\vec{x}_k) - \vec{\phi}_{1j}^T \vec{x}_k$, it follows that
\begin{align}
    \bar{\varphi}_j(\vec{x}_k) \left(\vec{\phi}_{j}^T \vec{x}_k - \bar{\varphi}_j(\vec{x}_k) \right) \geq 0
\end{align}
If we consider $\omega_j \geq 0$, $\Phi = \begin{bmatrix} \vec{\phi}_1 & \dots & \vec{\phi}_s \end{bmatrix}^T$, $\Phi_1 = \begin{bmatrix} \vec{\phi}_{11} & \dots & \vec{\phi}_{1s} \end{bmatrix}^T$, $\Omega = \operatorname{diag}(\begin{bmatrix} \omega_1 & \dots & \omega_s\end{bmatrix})$, and $\bar{\vec{\varphi}}^T(\vec{x}_k) = \begin{bmatrix} \bar{\varphi}_1(\vec{x}_k) & \dots & \bar{\varphi}_s(\vec{x}_k) \end{bmatrix}$, it follows that
{\footnotesize \begin{align}
   \begin{bmatrix} \bar{\varphi}_1(\vec{x}_k) & \bar{\varphi}_2(\vec{x}_k) & \dots & \bar{\varphi}_s(\vec{x}_k) \end{bmatrix} \begin{bmatrix} \omega_1 & 0 & \dots & 0 \\ 0 & \omega_2 & \dots 0 \\
    \vdots & \vdots & \ddots & \vdots \\
    0 & 0 & \dots & \omega_s \end{bmatrix} \left(\begin{bmatrix} \vec{\phi}_1^T \\ \vec{\phi}_2^T \\ \vdots \\ \vec{\phi}_s^T\end{bmatrix}\vec{x}_k - \begin{bmatrix} \bar{\varphi}_1(\vec{x}_k) \\ \bar{\varphi}_2(\vec{x}_k) \\ \vdots \\ \bar{\varphi}_s(\vec{x}_k) \end{bmatrix}\right) \geq 0
\end{align}}
which can be compactly written as
\begin{align}
    \bar{\vec{\varphi}}^T(\vec{x}_k) \Omega \left(\Phi \vec{x}_k - \bar{\vec{\varphi}}(\vec{x}_k) \right) \geq 0. \label{eq:S_procedure_sector}
\end{align}
This relationship shall prove useful to introduce extra information into our LMI conditions by means of the \emph{S-procedure}.

Consider the transformation of the uncertain nonlinearity that arises from adding and subtracting $G(\vec{x}_k)\Phi_1 \vec{x}_k$ to the open-loop system dynamics in \eqref{eq:ol_estimation}:
\begin{align}
    \vec{x}_{k+1} &= A(\vec{x}_k)\vec{x_k} + B(\vec{x}_k)\vec{u}_k + G(\vec{x}_k) \hat{\vec{\varphi}}(\vec{x}_k) \\&\quad+ G(\vec{x}_k) \left(\tilde{\Theta}_k^T\vec{h}(\vec{x}_k) + \vec{\varepsilon}(\vec{x}_k)\right) +  G(\vec{x}_k)\Phi_1 \vec{x}_k - G(\vec{x}_k)\Phi_1 \vec{x}_k,\\
    \vec{x}_{k+1} &= (A(\vec{x}_k) + G(\vec{x}_k)\Phi_1)\vec{x_k} + B(\vec{x}_k)\vec{u}_k + G(\vec{x}_k) \bar{\vec{\varphi}}(\vec{x}_k) \\&\quad+ G(\vec{x}_k) \left(\tilde{\Theta}_k^T\vec{h}(\vec{x}_k) + \vec{\varepsilon}(\vec{x}_k)\right).
\end{align}

Similarly to \eqref{eq:control_law}, we will employ an adaptive control law of the form
\begin{align}
    \vec{u}_k &= K(\vec{x}_k) \vec{x}_k - F(\vec{x}_k) \bar{\vec{\varphi}}(\vec{x}_k), \label{eq:control_law2}
\end{align}
which can be shown to lead to the closed loop dynamics
\begin{align}
    \vec{x}_{k+1} =& (A(\vec{x}_k) + G(\vec{x}_k)\Phi_1 + B(\vec{x}_k) K(\vec{x}_k))\vec{x}_k+\Xi_1^T(\vec{x}_k) \vec{\zeta}_k\\
    &+(G(\vec{x}_k) - B(\vec{x}_k)F(\vec{x}_k)) \bar{\vec{\varphi}}(\vec{x}_k)
\end{align}
with $\vec{\zeta}_k^T = \begin{bmatrix}\vec{\varepsilon}^T(\vec{x}_k) & h_{1k} \tilde{\vec{\theta}}_{1k} & \dots & h_{rk} \tilde{\vec{\theta}}_{rk} \end{bmatrix}$.

If we also consider a PDC control law, this case will lead to a control law of the form
\begin{align}
    \vec{u}_k &= \sum_{i=1}^r \mu_i(\vec{x}_k) \left(K_i \vec{x}_k - F_i \bar{\vec{\varphi}}(\vec{x}_k)\right), \label{eq:PDC2}
\end{align}
the closed loop system can be represented as
\begin{align}
        \vec{x}_{k+1} = \sum_{i=1}^{r_m} \sum_{j=1}^{r_m}\mu_i(\vec{x}_k) \mu_j(\vec{x}_k) \left(\bar{A}_{c_{ij}}\vec{x}_k +G_{c_{ij}} \bar{\vec{\varphi}}(\vec{x}_k)+\Xi_{1_i}(\vec{x_k}) \vec{\zeta}_k\right)
\end{align}
with $G_{c_{ij}}=\left(G_i - B_i F_j \right)$ and $\bar{A}_{c_{ij}}=A_i + G_i \Phi_1 + B_i K_j$.
Similarly to what was done for the matched uncertainty case, we consider a fuzzy Lyapunov function for $V_{x_k}$ of the form
\begin{align}
    V_{x_k} = \vec{x}_k^T \left(\sum_{i=1}^{r_m} \mu_i(\vec{x}_k) P_i\right)^{-1} \vec{x}_k = \vec{x}_k^T P^{-1}(\vec{x}_k) \vec{x}_k,
\end{align}
and make use of the conditions from Lemma \ref{lem:LyapX} in order to derive our LMI conditions. As such, we know that
\begin{align}
    0 < P_i \leq \frac{1}{\lambda} I, \quad \forall i
\end{align}
is sufficient to ensure that $V_{x_k} \geq \lambda \vec{x}_k^T \vec{x_k}$.

For the second set of conditions from the Lemma, it suffices that
\begin{align}
    &\vec{x}_{k+1}^T P^{-1}(\vec{x}_{k+1}) \vec{x}_{k+1} - (1-\eta) \vec{x}_k^T P^{-1}(\vec{x}_k) \vec{x}_k - \nu \vec{\zeta}_k^T \vec{\zeta}_k \leq  0 \\
    &\begin{bmatrix} \vec{x}_k \\ \bar{\vec{\varphi}} \\ \zeta_k \end{bmatrix}^T \Gamma(\vec{x}_k)\begin{bmatrix} \vec{x}_k \\ \bar{\vec{\varphi}} \\ \zeta_k \end{bmatrix} \leq 0
\end{align}
with 
\begin{align}
    \Gamma(\vec{x}_k)\! =\! \begin{bmatrix} -(1-\eta)P^{-1}(\vec{x}_k)  & 0 & 0\\ \ast & 0 & 0 \\ \ast & \ast & -\nu I\end{bmatrix} + \begin{bmatrix} \bar{A}^T_c(\vec{x}_k) \\ G_c^T(\vec{x_k}) \\ \Xi_1(\vec{x_k}) \end{bmatrix} P^{-1}(\vec{x}_{k+1}) \begin{bmatrix} \bar{A}^T_c(\vec{x}_k) \\ G_c^T(\vec{x}_k) \\ \Xi_1(\vec{x_k}) \end{bmatrix}^T,
\end{align}
$\bar{A}_c(\vec{x}_k) = A(\vec{x}_k) + G(\vec{x}_k)\Phi_1 + B(\vec{x}_k) K(\vec{x}_k)$ and $G_c=G(\vec{x}_k) - B(\vec{x}_k) F(\vec{x}_k)$. By making use of the conditions in \eqref{eq:S_procedure_sector} and employing the S-procedure, the following set of matrix inequalities is sufficient to guarantee this quadratic form
\begin{align} 
    &\begin{bmatrix} -(1-\eta)P^{-1}(\vec{x}_k)  & 0 & 0\\ \ast & 0 & 0 \\ \ast & \ast & -\nu I\end{bmatrix} + \begin{bmatrix} \bar{A}^T_c(\vec{x}_k) \\ G_c^T(\vec{x}_k)\\ \Xi_1(\vec{x}_k) \end{bmatrix} P^{-1}(\vec{x}_{k+1}) \begin{bmatrix} \bar{A}^T_c(\vec{x}_k) \\ G_c^T(\vec{x}_k) \\ \Xi_1(\vec{x_k}) \end{bmatrix}^T \\
    &+ \begin{bmatrix} 0 \\ I \\ 0 \end{bmatrix} \Omega \begin{bmatrix} \Phi^T \\ -I \\ 0 \end{bmatrix}^T + \begin{bmatrix} \Phi^T \\ -I \\ 0 \end{bmatrix} \Omega \begin{bmatrix} 0 \\ I \\ 0 \end{bmatrix}^T \leq 0
\end{align}

From Schur's complement, these LMIs are equivalent to
\begin{align}
    \begin{bmatrix} -(1-\eta)P^{-1}(\vec{x}_k)  &  \Phi^T \Omega & 0 & \bar{A}_c^T\\ \ast & -2\Omega & 0 & G_c^T(\vec{x}_k)\\ \ast & \ast & -\nu I & \Xi_1(\vec{x}_k) \\ \ast & \ast & \ast & -P(\vec{x}_{k+1}) \end{bmatrix} \leq 0
\end{align}
which, by multiplying it on the left by $\diag{M^T, \Omega^{-1}, I, I}$ and on the right by $\diag{M, \Omega^{-1}, I, I}$, making use of $M^T P^{-1}(\vec{x}_k) M \geq M^T + M - P(\vec{x}_k)$, and defining $\bar{\Omega} = \Omega^{-1}$, $\bar{K}(\vec{x}_k) = K(\vec{x}_k)M$, and $\bar{F}(\vec{x}_k) = F(\vec{x}_k)\bar{\Omega}$, leads to
\begin{align}
    \begin{bmatrix} -(1-\eta)\left(M^T + M - P(\vec{x}_k)\right)  &  M^T \Phi^T & 0 & \Xi_3(\vec{x_k})\\ \ast & -2\bar{\Omega} & 0 & \Xi_2(\vec{x_k})\\ \ast & \ast & -\nu I & \Xi_1(\vec{x_k}) \\ \ast & \ast & \ast & -P(\vec{x}_{k+1}) \end{bmatrix} \leq 0
\end{align}
with $\Xi_2(\vec{x_k})=\bar{\Omega} G^T(\vec{x}_k) - \bar{F}^T(\vec{x}_k) B^T(\vec{x}_k)$ and $\Xi_3(\vec{x_k})=M^T A^T(\vec{x}_k) + M^T \Phi_1^T G^T(\vec{x}_k) + \bar{K}^T(\vec{x}_k) B^T(\vec{x}_k)$.
If we employ a similar strategy to the one used for the matched uncertainty case, the $\vec{1}_{r+1} \otimes G^T(\vec{x}_k)$ can be simplified, leading to
\begin{align}
    \begin{bmatrix} -(1-\eta)\left(M^T + M - P(\vec{x}_k)\right) & M^T\Phi^T &  \Xi_3(\vec{x_k}) & 0 \\ 
    \ast & -2\bar{\Omega} & \Xi_2(\vec{x_k}) & 0 \\ 
    \ast & \ast & -P(\vec{x}_{k+1}) 
 & G(\vec{x}_k) \\ \ast & \ast & \ast & -\frac{\nu}{r+1} I \end{bmatrix} \leq 0.
\end{align}

Considering these steps, Theorem \ref{thm:controle_setor} summarizes the synthesis results for the case of sector-bounded uncertainties.
\begin{thm} \label{thm:controle_setor}
    Consider the nonlinear discrete-time system in \eqref{eq:sys_qLPV}, satisfying Assumption \ref{asm:sector} and admitting the TS fuzzy model description in \eqref{eq:sys_TS}. If there exist symmetric matrices $P_i \in \mathds{R}^{n \times n}, \; i \in [1, \dots, r_m]$, a diagonal positive definite matrix $\bar{\Omega} \in \mathds{R}^{s \times s}$, matrices $M \in \mathds{R}^{n \times n}$, $\bar{K}_i \in \mathds{R}^{m \times n}$, and $\bar{F}_i \in \mathds{R}^{m \times s}$, with $i \in [1, \dots, r_m]$ and a scalar $\nu$ such that
    \begin{align}
        0 < P_{i} &\leq \dfrac{1}{\lambda} I, \quad \forall i \\
        Q_{ii\ell} &\leq 0, \quad \forall i, \ell \\
        Q_{ij\ell} + Q_{ji\ell} & \leq 0, \quad \forall \ell, i, j> i
    \end{align}
    with
    \begin{align}
        Q_{ij\ell} = \begin{bmatrix} -(1-\eta)\left(M^T + M - P_i\right)  &  M^T \Phi^T & \Xi_{3_{ij}}  & 0 \\ \ast & -2\bar{\Omega} & \Xi_{2_{ij}} & 0 \\ \ast & \ast & -P_\ell & G_i \\ \ast & \ast & \ast & -\frac{\nu}{r+1} I \end{bmatrix},
    \end{align}   
    $\Xi_{3_{ij}}=M^T A^T_i + M^T \Phi_1^T G_i^T$ and $\Xi_{2_{ij}}=\bar{\Omega} G_i^T - \bar{F}_j^T B_i^T$, then the adaptive control law \eqref{eq:PDC2} (with $K_i = \bar{K}_i M^{-1}$ and $F_i = \bar{F}_i \bar{\Omega}^{-1}$) together with the adaptation law in \eqref{eq:cov}, \eqref{eq:xi_proj}, and \eqref{eq:theta_proj}, ensures that the states $\vec{x}_k$ of the closed loop system are \emph{ultimately uniformly bounded} with bound \eqref{eq:UUB_bound}.
\end{thm}

\subsection{Norm-bounded Nonlinearity}
As a last case, we consider unknown nonlinearities that are bounded (for all times), as summarized in the following assumption.
\begin{asm}[Norm-bounded Nonlinearity] \label{asm:norm}
    The unknown function $\vec{\varphi}(\vec{x}_k)$ satisfies
    \begin{align}
        \|\vec{\varphi}(\vec{x}_k)\|^2_2 \leq \psi, \quad \forall k.
    \end{align}
\end{asm}

As in the sector-bounded case, we consider that our approximator $\hat{\vec{\varphi}}(\vec{x}_k) = \hat{\Theta}_k^T \vec{h}(\vec{x}_k)$ needs to respect the same bound as the original nonlinearity. As such, we could require that
\begin{align}
    \|\hat{\Theta}_k^T \vec{h}(\vec{x}_k)\|_2^2 \leq \psi, \quad \forall k.
\end{align}
However, considering the local estimation approach we are employing in this paper, we instead opt to use the following sufficient bound
\begin{align}
    \|\hat{\vec{\theta}}_{ik}^T\|_2^2 &\leq \psi, \quad \forall k \\
    \hat{\vec{\theta}}_{ik} \hat{\vec{\theta}}_{ik}^T &\leq \psi, \quad \forall k.
\end{align}

In order to impose these constraints, the following projection modification is done to the adaptation law in \eqref{eq:mqr_local}
\begin{align}
     \vec{\xi}^T_{i(k+1)} &= \hat{\vec{\theta}}^T_{ik} + p_{i(k+1)} (G^T(\vec{x}_k) G(\vec{x}_k))^{-1} G^T(\vec{x}_k) h_{ik} \vec{e}^{(i)}_{k+1}  \label{eq:xi_proj2}\\
     \hat{\vec{\theta}}^T_{i(k+1)} &= \left\{\begin{array}{cl} \vec{\xi}^T_{i(k+1)}, & \textrm{if } \vec{\xi}_{i(k+1)} \vec{\xi}^T_{i(k+1)} \leq \psi, \\[0.12cm] 
     \sqrt{\psi} \dfrac{\vec{\xi}^T_{i(k+1)}}{\sqrt{\vec{\xi}_{i(k+1)}\vec{\xi}^T_{i(k+1)}}}, & \textrm{otherwise.} \end{array} \right. \label{eq:theta_proj2}
\end{align}
Since we consider that the original $\Theta$ parameters being estimated are inside this region, the use of this projection always leads to a decrease in the adaptation error. In addition, since our covariances are always described by scalars, this decrease directly corresponds to a decrease in the local adaptation Lyapunov function \eqref{eq:local_lyap_est}, and the terms introduced by the projection can be ignored for stability/convergence analysis.

In this case, the uncertainty cannot be completely compensated for by the controller (as it was in the previous two cases), and a term relating to $\psi$ will remain in the stability analysis. Therefore, we cannot directly employ Lemma \ref{lem:LyapX} and instead adapt it as follows.

\begin{lem} \label{lem:LyapX2}
    Consider that the unknown nonlinearity in system \eqref{eq:open_loop} satisfies Assumption \ref{asm:norm}. If there exists a control law of the form
    \begin{align}
        \vec{u}_k &= \left(\sum_{i=1}^r \mu_i(\vec{x}_k) K_i \vec{x}_k\right) - F(\vec{x}_k) \hat{\vec{\varphi}}(\vec{x}_k), \label{eq:PDC3}
    \end{align}
    such that there exists a Lyapunov function for the system's states, $V_{x_k}$, satisfying
    \begin{align}
        V_{x_k} &\geq \lambda \vec{x}_k^T \vec{x}_k \\
        V_{x_{k+1}} - V_{x_k} = \Delta V_x &\leq - \eta V_{x_k} + \nu_1 \left(\vec{\varepsilon}_k^T \vec{\varepsilon}_k+\sum_{i=1}^r h_{ik}^2 \tilde{\vec{\theta}}_{ik} \tilde{\vec{\theta}}_{ik}^T \right) + \nu_2 \breve{\vec{\varphi}}^T(\vec{x}_k) \breve{\vec{\varphi}}(\vec{x}_k)
    \end{align}
    with $\breve{\vec{\varphi}}(\vec{x}_k) = (G(\vec{x}_k)- B(\vec{x}_k) F(\vec{x}_k) )\hat{\vec{\varphi}}(\vec{x}_k)$, then the closed loop employing said control law with the adaptation law in \eqref{eq:cov}, \eqref{eq:xi_proj2}, and \eqref{eq:theta_proj2} ensures that the system states, $\vec{x}_k$ will be \emph{ultimately uniformly bounded} with the bound
    \begin{align}
        \|x_k\|_2^2 \leq \left(\nu_1 \delta \left(1 + \dfrac{4 \bar{p}^2}{\alpha^2} \right) + \nu_2 \left\|G(\vec{x}_k)- B(\vec{x}_k) F(\vec{x}_k)\right\|_2^2\psi \right) \left(\dfrac{1 + \eta}{\eta \lambda} \right) \label{eq:UUB_bound2}
    \end{align}
    with $\bar{p} \geq \alpha$ the maximum value chosen for the initial condition of the covariances, $\vec{\varepsilon}_{k}^T \vec{\varepsilon}_{k} \leq \delta$, and $\vec{\varepsilon}_{ik}^T \vec{\varepsilon}_{ik} \leq \delta$.
\end{lem}
\begin{proof}
    Consider a Lyapunov function for the whole closed loop system (states and estimation error) given by
    \begin{align}
        V_k = V_{x_k} + \frac{\bar{p} \nu_1}{\gamma \alpha} \sum_{i = 1}^r V_{\theta_{ik}}
    \end{align}
    following similar steps to those in the proof of Lemma \ref{lem:LyapX}, we have that
    \begin{align}
         \Delta V \leq - \eta V_{x_k} + \bar{\nu}_1 \delta + \bar{\nu}_2 \psi,
    \end{align}
    with $\bar{\nu}_1 = \nu_1 \left(1 + \dfrac{\bar{p}^2}{(1 - \gamma)\gamma\alpha^2} \right)$, and $\bar{\nu}_2 = \nu_2 \left\|G(\vec{x}_k)- B(\vec{x}_k) F(\vec{x}_k)\right\|_2^2$.
    This is guaranteed to be less than or equal to zero as long as $V_{x_k} \geq \frac{\bar{\nu}_1\delta + \bar{\nu}_2 \psi}{\eta}$, driving the states into the region $V_{x_k} \leq \frac{\bar{\nu}_1\delta + \bar{\nu}_2 \psi}{\eta}$. From within this region, the maximum change that can happen to the Lyapunov function is $\bar{\nu}_1 \delta + \bar{\nu}_2 \psi$, which is outside of the region and will lead to a decrease of the Lyapunov function once again. Hence, the system is bounded such that
    \begin{align}
        \lambda \vec{x}_k^T \vec{x}_k \leq V_{x_k} \leq \frac{\eta + 1}{\eta}\left(\bar{\nu}_1\delta + \bar{\nu}_2 \psi\right),
    \end{align}
    and the bounded region in \eqref{eq:UUB_bound2} is found by considering $\gamma = 0.5$.
\end{proof}

By making use of the control law in \eqref{eq:PDC3}, the closed loop dynamics will be given by
\begin{align}
    \vec{x}_{k+1} &= (A(\vec{x}_k) + B(\vec{x}_k) K(\vec{x}_k))\vec{x}_k +G_c(\vec{x}_k) \hat{\vec{\varphi}}(\vec{x}_k) +\Xi_1(\vec{x}_k) \vec{\zeta}_k,
\end{align}
with $G_c(\vec{x}_k) = G(\vec{x}_k)- B(\vec{x}_k) F(\vec{x}_k)$. By following similar steps to those presented for Theorems \ref{thm:controle_casado} and \ref{thm:controle_setor}, but employing Lemma \ref{lem:LyapX2} instead of Lemma \ref{lem:LyapX}, we arrive at the synthesis conditions presented in the following theorem.

\begin{thm} \label{thm:controle_norma}
    Consider the nonlinear discrete-time system in \eqref{eq:sys_qLPV}, satisfying Assumption \ref{asm:norm} and admitting the TS fuzzy model description in \eqref{eq:sys_TS}. If there exist symmetric matrices $P_i \in \mathds{R}^{n \times n}, \; i \in [1, \dots, r_m]$, matrices $M \in \mathds{R}^{n \times n}$, $\bar{K}_i \in \mathds{R}^{m \times n}$ with $i \in [1, \dots, r_m]$ and scalars $\nu_1$ and $\nu_2$ such that
    \begin{align}
        0 < P_{i} &\leq \dfrac{1}{\lambda} I, \quad \forall i \\
        Q_{ii\ell} &\leq 0, \quad \forall i, \ell \\
        Q_{ij\ell} + Q_{ji\ell} & \leq 0, \quad \forall \ell, i, j> i
    \end{align}
    with
    \begin{align}
        Q_{ij\ell} = \begin{bmatrix} -(1-\eta)\left(M^T + M - P_i\right)  &  0 & M^T A^T_i + \bar{K}_j^T B_i^T & 0 \\ \ast & -\nu_2 I & I & 0 \\ \ast & \ast & -P_\ell & G_i \\ \ast & \ast & \ast & -\frac{\nu_1}{r+1} I \end{bmatrix}
    \end{align}
    then the adaptive control law \eqref{eq:PDC3} (with $K_i = \bar{K}_i M^{-1}$) together with the adaptation law in \eqref{eq:cov}, \eqref{eq:xi_proj2}, and \eqref{eq:theta_proj2} ensure that the states $\vec{x}_k$ of the closed loop system are \emph{ultimately uniformly bounded} with bound \eqref{eq:UUB_bound2}.
\end{thm}

Note that the conditions in Theorem \ref{thm:controle_norma} guarantee that the adaptive control law will be ultimately uniformly bounded regardless of the gains $F(\vec{x}_k)$, as long as $\left\|G(\vec{x}_k)- B(\vec{x}_k) F(\vec{x}_k)\right\|_2$ is bounded. If $\hat{\vec{\varphi}}(\vec{x}_k)$ were simply a measured disturbance, the conditions in Theorem \ref{thm:controle_norma} could easily be modified to directly search for a set of gains $F(\vec{x}_k)$ that minimized this norm or its effect on the system dynamics. 

Unlike the previous cases, the uncertainty cannot be completely compensated, and instead we will focus our attention on compensating a desired output 
\begin{equation}
    \vec{y}_k = C(\vec{x}_k) \vec{x}_k \label{eq:output}
\end{equation}
and ensuring that it rejects the unknown nonlinearity (instead of the states). In order to do so, we make use of the following theorem.

\begin{thm}[Adapted from Theorem 1.35 in \cite{Huang2004}] \label{thm:Francis_regulator}
    Consider the nonlinear system described by \eqref{eq:sys_qLPV}. If there exist matrix functions $H(\vec{x}_k) \in \mathds{R}^{n \times s}$ and $U(\vec{x}_k) \in \mathds{R}^{m \times s}$ that solve
    \begin{equation}
        \begin{bmatrix}
            A(\vec{x}_k) - I & B(\vec{x}_k) \\ C(\vec{x}_k) & 0 
        \end{bmatrix}
        \begin{bmatrix}
            H(\vec{x}_k) \\ U(\vec{x}_k)
        \end{bmatrix} = 
        \begin{bmatrix}
            -G(\vec{x}_k) \\ 0
        \end{bmatrix} \label{eq:Francis_regulator}
    \end{equation}
    then control law \eqref{eq:PDC3}, with the feedforward term
    \begin{equation}
        F(\vec{x}_k) = K(\vec{x}_k)H(\vec{x}_k) - U(\vec{x}_k), \label{eq:ff_norm_bound}
    \end{equation}
    and $K(\vec{x}_k) = \sum_{i=1}^r \mu_i(\vec{x}_k) K_i$, \textbf{approximately} ensures that the desired output \eqref{eq:output} would not be affected by the \textbf{ideal} approximator $\vec{\varphi}^\ast(\vec{x}_k) = \Theta^T\vec{h}(\vec{x}_k)$.
\end{thm}

\begin{proof}
    The nonlinear system in \eqref{eq:sys_qLPV} can be rewritten, using \citep[Proposition 7.1.9]{Bernstein2009}, as
    \begin{align}
        \vec{x}_{k+1} &= A(\vec{x}_k)\vec{x}_k + B(\vec{x}_k)\vec{u}_k + \operatorname{vec}(G(\vec{x}_k) \Theta^T \vec{h}(\vec{x}_k)) + G(\vec{x}_k) \vec{\varepsilon}_k, \\
        \vec{x}_{k+1} &= A(\vec{x}_k)\vec{x}_k + B(\vec{x}_k)\vec{u}_k + (\vec{h}^T(\vec{x}_k) \otimes G(\vec{x}_k))\operatorname{vec}(\Theta^T) + G(\vec{x}_k) \vec{\varepsilon}_k.
    \end{align}
    We define the following change of coordinates
    \begin{align}
        \bar{\vec{x}}_k = \vec{x}_k - \bar{H}(\vec{x}_k)\operatorname{vec}(\Theta^T), && \bar{\vec{u}}_k = \vec{u}_k - \bar{U}(\vec{x}_k)\operatorname{vec}(\Theta^T), \\
        \vec{x}_k = \bar{\vec{x}}_k + \bar{H}(\vec{x}_k)\operatorname{vec}(\Theta^T), && \vec{u}_k = \bar{\vec{u}}_k + \bar{U}(\vec{x}_k)\operatorname{vec}(\Theta^T),
    \end{align}
    which allows us to rewrite the dynamics as
    \begin{align}
        \bar{\vec{x}}_{k+1} & = A(\vec{x}_k)\bar{\vec{x}}_k + B(\vec{x}_k)\bar{\vec{u}}_k + G(\vec{x}_k) \vec{\varepsilon}_k - (\bar{H}(\vec{x}_{k+1}) - \bar{H}(\vec{x}_k))\operatorname{vec}(\Theta^T)\\&+ (\left(A(\vec{x}_k) - I\right)\bar{H}(\vec{x}_k) + B(\vec{x}_k)\bar{U}(\vec{x}_k) + \vec{h}^T(\vec{x}_k) \otimes G(\vec{x}_k))\operatorname{vec}(\Theta^T), \\
        \vec{y}_k &= C(\vec{x}_k) \bar{\vec{x}}_k + C(\vec{x}_k)\bar{H}(\vec{x}_k)\operatorname{vec}(\Theta^T).
    \end{align}

    If we are able to find $\bar{H}(\vec{x}_k)$ and $\bar{U}(\vec{x}_k)$ such that 
    \begin{equation}
        \begin{bmatrix}
            A(\vec{x}_k) - I & B(\vec{x}_k) \\ C(\vec{x}_k) & 0 
        \end{bmatrix}
        \begin{bmatrix}
            \bar{H}(\vec{x}_k) \\ \bar{U}(\vec{x}_k)
        \end{bmatrix} = 
        \begin{bmatrix}
            -\vec{h}^T(\vec{x}_k)\otimes G(\vec{x}_k) \\ 0
        \end{bmatrix} \label{eq:francis_transformed}
    \end{equation}
    and consider that $\bar{H}(\vec{x}_{k+1}) - \bar{H}(\vec{x}_k) \approx 0$, the dynamics become
    \begin{align}
        \bar{\vec{x}}_{k+1} & \approx A(\vec{x}_k)\bar{\vec{x}}_k + B(\vec{x}_k)\bar{\vec{u}}_k + G(\vec{x}_k) \vec{\varepsilon}_k, \\
        \vec{y}_k &= C(\vec{x}_k) \bar{\vec{x}}_k.
    \end{align}

    Note, however, that the solution to equation \eqref{eq:francis_transformed} can be related to the solution to equation \eqref{eq:Francis_regulator} with
    \begin{align}
        \bar{H}(\vec{x}_k) = \vec{h}^T(\vec{x}_k) \otimes H(\vec{x}_k), && \bar{U}(\vec{x}_k) = \vec{h}^T(\vec{x}_k) \otimes U(\vec{x}_k)
    \end{align}
    and the smaller equation \eqref{eq:Francis_regulator} can be used instead.
    
    By considering control law \eqref{eq:PDC3}, we arrive at
    \begin{align}
        \vec{u}_k &= K(\vec{x}_k)\vec{x}_k - F(\vec{x}_k)\Theta^T\vec{h}(\vec{x}_k), \\
        \bar{\vec{u}}_k&=K(\vec{x}_k)\bar{\vec{x}}_k + K(\vec{x}_k)\bar{H}(\vec{x}_k)\operatorname{vec}(\Theta^T)- \bar{U}(\vec{x}_k)\operatorname{vec}(\Theta^T) \\
        &- \left(\vec{h}^T(\vec{x}_k) \otimes F(\vec{x}_k)\right)\operatorname{vec}(\Theta^T), \\
        \bar{\vec{u}}_k&=K(\vec{x}_k)\bar{\vec{x}}_k + (1 \otimes K(\vec{x}_k))(\vec{h}^T(\vec{x}_k) \otimes H(\vec{x}_k))\operatorname{vec}(\Theta^T) \\&- (\vec{h}^T(\vec{x}_k) \otimes U(\vec{x}_k))\operatorname{vec}(\Theta^T) - \left(\vec{h}^T(\vec{x}_k) \otimes F(\vec{x}_k)\right)\operatorname{vec}(\Theta^T), \\
        \bar{\vec{u}}_k&=K(\vec{x}_k)\bar{\vec{x}}_k + (\vec{h}^T(\vec{x}_k) \otimes\left(K(\vec{x}_k)H(\vec{x}_k) - U(\vec{x}_k) - F(\vec{x}_k)\right))\operatorname{vec}(\Theta^T).
    \end{align}
    Since we are considering the gains given by \eqref{eq:ff_norm_bound}, it follows that
    \begin{equation}
        \bar{\vec{u}}_k=K(\vec{x}_k)\bar{\vec{x}}_k,
    \end{equation}
    and the closed loop system can be approximately given by
    \begin{align}
        \bar{\vec{x}}_{k+1} & \approx (A(\vec{x}_k) + B(\vec{x}_k) K(\vec{x}_k))\bar{\vec{x}}_k + G(\vec{x}_k) \vec{\varepsilon}_k, \\
        \vec{y}_k &= C(\vec{x}_k) \bar{\vec{x}}_k,
    \end{align}
    which concludes the proof.
\end{proof}

\begin{rem}
    One might be tempted, from the conditions in Theorem \ref{thm:Francis_regulator}, to simply design a control gain $K(\vec{x}_k)$ that rejects the approximation error $\vec{\varepsilon}_k$ instead of using the conditions in Theorem \ref{thm:controle_norma} (which ensure ultimate boundedness for any bounded feedforward term $F(\vec{x}_k)$). However, note that the conditions in Theorem \ref{thm:Francis_regulator} are only exact when we employ the ideal parameters $\Theta^T$ (since we are not taking into account the dynamics of the estimated parameters), and when $\bar{H}(\vec{x}_{k+1}) - \bar{H}(\vec{x}_k) = 0$ (which is quite hard to ensure since, even if we are able to find a solution $H(\vec{x}_k) = H$ to equation \eqref{eq:Francis_regulator} that is independent of the states, $\bar{H}(\vec{x}_k) = \vec{h}^T(\vec{x}_k) \otimes H$, and that would require very slow dynamics for the approximator membership functions to be true). As such, the use of Theorem \ref{thm:controle_norma} in conjunction with Theorem \ref{thm:Francis_regulator} is still the safest approach.
\end{rem}

\section{Numerical examples}

\subsection{Matched Nonlinearity}

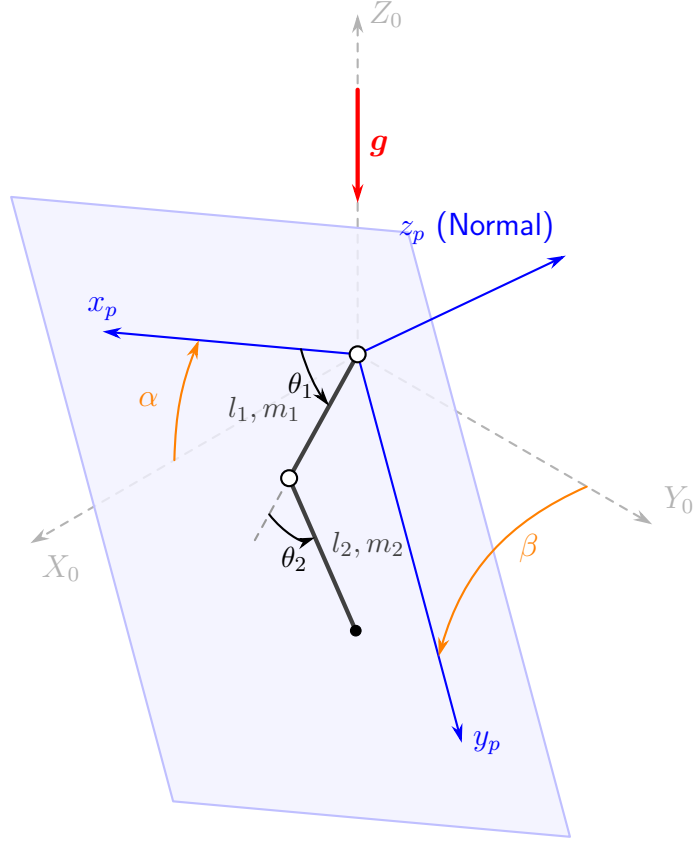
\begin{figure}[htb!]
    \centering
    \begin{tikzpicture}[
        scale=1.0, 
        x={(-0.866cm, -0.5cm)},
        y={(0.866cm, -0.5cm)},
        z={(0cm, 1cm)},
        >={Stealth[length=2.5mm, width=1.5mm]},
        font=\sffamily,
        line join=round,
        line cap=round
        ]
    
        \coordinate (O) at (0,0,0);
    
        \draw[->, thick, dashed, gray!60] (O) -- (5.0,0,0) node[below right] {$X_0$};
        \draw[->, thick, dashed, gray!60] (O) -- (0,4.5,0) node[above right] {$Y_0$};
        \draw[->, thick, dashed, gray!60] (O) -- (0,0,4.5) node[right] {$Z_0$};
    
        \draw[->, ultra thick, red] (0,0,3.5) -- (0,0,2.0) node[midway, right] {$\vec{g}$};
    
        \coordinate (xp) at (0.866, 0, 0.5);
        \coordinate (yp) at (0.3535, 0.707, -0.612);
        \coordinate (zp) at (-0.3535, 0.707, 0.612);
    
        \coordinate (P1) at ($ -1.5*(xp) - 1.5*(yp) $);
        \coordinate (P2) at ($ 5.5*(xp) - 1.5*(yp) $);
        \coordinate (P3) at ($ 5.5*(xp) + 5.5*(yp) $);
        \coordinate (P4) at ($ -1.5*(xp) + 5.5*(yp) $);
        \fill[blue!5, opacity=0.8] (P1) -- (P2) -- (P3) -- (P4) -- cycle;
        \draw[blue!25, thick] (P1) -- (P2) -- (P3) -- (P4) -- cycle;
    
        \draw[->, thick, blue] (O) -- ($ 4.5*(xp) $) node[above] {$x_p$};
        \draw[->, thick, blue] (O) -- ($ 4.5*(yp) $) node[right] {$y_p$};
        \draw[->, thick, blue] (O) -- ($ 3*(zp) $) node[above left] {$z_p$ (Normal)};
    
        \draw[->, thick, orange] (2.8,0,0) to[bend left=10] node[midway, left=4pt] {$\alpha$} ($ 2.8*(xp) $);
        \draw[->, thick, orange] (0,3.5,0) to[bend right=25] node[midway, right=6pt, yshift=2pt] {$\beta$} ($ 3.5*(yp) $);
    
        \begin{scope}[x={(xp)}, y={(yp)}]
            \def\Lone{2.4}
            \def\Ltwo{1.8}
            \def\thOne{40}
            \def\thTwo{65}
    
            \coordinate (J1) at (\thOne:\Lone);
            \coordinate (J2) at ($ (J1) + (\thOne+\thTwo:\Ltwo) $);
    
            \draw[dashed, thick, gray!80] (J1) -- ($ (J1) + (\thOne:1.2) $);
    
            \draw[ultra thick, black!75] (O) -- (J1) node[pos=0.65, above left, font=\small] {$l_1, m_1$};
            \draw[ultra thick, black!75] (J1) -- (J2) node[pos=0.60, above right, font=\small] {\!\!$l_2, m_2$};
    
            \draw[->, thick, black] (1.0, 0) arc (0:\thOne:1.0);
            \node[font=\small] at (\thOne/2:1.25) {$\theta_1$};
    
            \draw[->, thick, black] ($ (J1) + (\thOne:0.7) $) arc (\thOne:\thOne+\thTwo:0.7);
            \node[font=\small] at ($ (J1) + (\thOne+\thTwo/2:1.0) $) {$\theta_2$};
        \end{scope}
    
        \filldraw[fill=white, draw=black, thick] (O) circle (3pt);
        \filldraw[fill=white, draw=black, thick] (J1) circle (3pt);
        \filldraw[fill=black] (J2) circle (2pt);
    
    \end{tikzpicture}
    \label{fig:manipulator}
    \caption{Planar manipulator on a rotated plane. Given an initial plane that is parallel to the earth (represented by the $X_0Y_0Z_0$ axis), the manipulator's plane is rotated by $\alpha$ around the $y$-axis, followed by a $\beta$ rotation around the rotated $x$-axis.}
\end{figure}

Consider the planar two link robotic manipulator on a rotated plane illustrated by Figure \ref{fig:manipulator}. Considering that the link's lengths are given by $l_1$ and $l_2$, the masses by $m_1$ and $m_2$, and the joints' rotations by $\theta_1$ and $\theta_2$. Given an initial plane that is parallel to the earth (represented by the $X_0Y_0Z_0$ axis), the manipulator's plane is rotated by $\alpha$ around the $y$-axis, followed by a $\beta$ rotation around the rotated $x$-axis.

This manipulator's dynamics can be represented using the Euler-Lagrange equations as
\begin{align}
    M(\theta_2)\begin{bmatrix}
        \ddot{\theta}_1 \\ \ddot{\theta}_2
    \end{bmatrix} + C(\theta_1,\theta_2,\dot{\theta}_2)\begin{bmatrix}
        \dot{\theta}_1 \\ \dot{\theta}_2
    \end{bmatrix} + \vec{g}(\theta_1,\theta_2) = \begin{bmatrix} u_1 \\ u_2 \end{bmatrix},
\end{align}
with $u_1$ and $u_2$ the input torques applied to each joint, and
\begin{equation}\begin{split}
    M(\theta_2)\!\!=\!\!\begin{bmatrix}
        \left(\frac{1}{3}m_1 + m_2\right)l_1^2 + \frac{1}{3}m_2l_2^2+m_2l_1l_2\cos(\theta_2) & \frac{1}{3}m_2l_2^2+\frac{1}{2}m_2l_1l_2\cos(\theta_2) \\
        \frac{1}{3}m_2l_2^2+\frac{1}{2}m_2l_1l_2\cos(\theta_2) & \frac{1}{3}m_2l_2^2
    \end{bmatrix}, \\
    C(\theta_1,\theta_2,\dot{\theta}_2) = \begin{bmatrix} 
        -\frac{1}{2}m_2l_1l_2\sin(\theta_2)\dot{\theta}_2 & -\frac{1}{2}m_2l_1l_2\sin(\theta_2)\left(\dot{\theta}_1 + \dot{\theta}_2\right) \\
        \frac{1}{2}m_2l_1l_2\sin(\theta_2)\dot{\theta}_1 & 0
    \end{bmatrix}, \\
    \vec{g}(\theta_1,\theta_2) = M_g(\theta_1,\theta_2) \begin{bmatrix} \cos(\alpha)\sin(\beta) \\ \sin(\alpha)\end{bmatrix},\\
    M_g(\theta_1,\theta_2) = g \begin{bmatrix} 
    \left(\frac{m_1}{2}+m_2\right)l_1c_1+\frac{m_2l_2}{2}c_{12} & \left(\frac{m_1}{2}+m_2\right)l_1s_1 + \frac{m_2l_2}{2}s_{12} \\
    \frac{m_2l_2}{2}c_{12} & \frac{m_2l_2}{2} s_{12}
    \end{bmatrix},
    \end{split}
\end{equation}
with $g$ being the acceleration due to gravity, $c_1=\cos(\theta_1)$, $c_{12}=\cos(\theta_1+\theta_2)$, $s_1=\sin(\theta_1)$, and $s_{12}=\sin(\theta_1+\theta_2)$.

In this example, the $\alpha$ and $\beta$ angles are unknown, so gravity's effect on the system's dynamics is considered to be an unknown function and $\vec{\varphi}(\theta_1,\theta_2) = \vec{g}(\theta_1,\theta_2)$. The system's dynamics can then be recast, in the traditional state-space format, as
\begin{align}
    \underbrace{\begin{bmatrix} \dot{\theta}_1  \\  \dot{\theta}_2 \\ \ddot{\theta}_1 \\ \ddot{\theta}_2 \end{bmatrix}}_{\dot{\vec{x}}} = \underbrace{\begin{bmatrix} 0_{2\times2} & I_{2\times 2} \\ 0_{2\times 2} & -M^{-1}(\theta_2)C(\theta_1,\theta_2,\dot{\theta}_2)\end{bmatrix}}_{A(\vec{x})} \underbrace{\begin{bmatrix} \theta_1  \\  \theta_2 \\ \dot{\theta}_1 \\ \dot{\theta}_2 \end{bmatrix}}_{\vec{x}} + \underbrace{\begin{bmatrix} 0_{2\times 2} \\ M^{-1}(\theta_2)\end{bmatrix}}_{B(\vec{x})=-G(\vec{x})} (\vec{u} - \vec{\varphi}(\theta_1,\theta_2)).
\end{align}

By noting that the inverse of $M(\theta_2)$ can be written as
\begin{align}
    &M^{-1}(\theta_2) = \\ &\quad \dfrac{1}{\det(M(\theta_2))} \begin{bmatrix}
        \frac{1}{3}m_2l_2^2 & -\frac{1}{3}m_2l_2^2-\frac{1}{2}m_2l_1l_2c_2 \\
        -\frac{1}{3}m_2l_2^2-\frac{1}{2}m_2l_1l_2c_2 & (\frac{1}{3}m_1+m_2)l_1^2+\frac{1}{3}m_2l_2^2+m_2l_1l_2c_2
    \end{bmatrix}, \\
    &\det(M(\theta_2)) = \frac{1}{9}m_1m_2l_1^2l_2^2 + \left(\frac{1}{3}-\frac{1}{4}\cos^2(\theta_2)\right)m_2^2l_1^2l_2^2,
\end{align}
with $c_2=\cos(\theta_2)$, and, by making use of the sector-nonlinearity approach \cite{Tanaka2004} with $\frac{1}{\det(M(\theta_2))}$, $\cos(\theta_2)$, $\sin(\theta_2)\dot{\theta}_2$, and $\sin(\theta_2)\dot{\theta}_1$ as premise variables, we arrive at a TS model with 16 vertices. 

The conditions in Theorem \ref{thm:controle_casado} lead to a set of controller gains that aim to reject the effect of the unknown gravity on the closed loop system since it always enters the dynamics in the same direction as the approximation error (which is the actual purpose of the conditions in the theorem). In that regard, in this example, we can easily compare the behaviour of the proposed adaptive controller against a robust controller that aims to mitigate the effect of gravity on the manipulator (by simply not using the estimation generated by the adaptation law). 

Note, however, that, unlike in the adaptive case in which the approximation errors are always bounded within a compact region of the state space, this might not necessarily be true for the unknown nonlinearity. These unknown nonlinearities are arbitrary functions of the states and feed the state information back into the closed loop, possibly leading to instability. Thus, in most cases, Theorem \ref{thm:controle_casado} \emph{should not} be employed to design a controller for nonlinear systems to be robust against the unknown nonlinearities. In this particular example, though, we know that this unknown nonlinearity is always bounded (in the whole state space) as it represents the unknown effect of gravity upon the manipulator.

For this example, we consider that $m_1 = 2 kg$, $m_2 = 2 kg$, $l_1 = 0.3m$, $l_2 = 0.3m$, and $g = 9.78 m/s^2$. For the TS fuzzy model used in the synthesis, we consider that $\dot{\theta}_1 \in [-\pi/2, \pi/2]$ and $\dot{\theta}_2 \in [-\pi/2, \pi/2]$. In order to find a discrete-time model for the system, we employ a forward Euler discretization with $\Delta t = 1 \times 10^{-2}s$ and make use of the conditions in Theorem \ref{thm:controle_casado} to find a controller that minimizes the effect of the approximation error/unknown function on the closed loop, with $\eta = 1\times 10^{-2}$ and $\lambda = 1$. 

\begin{figure}[htb!]
    \centering
    \includegraphics{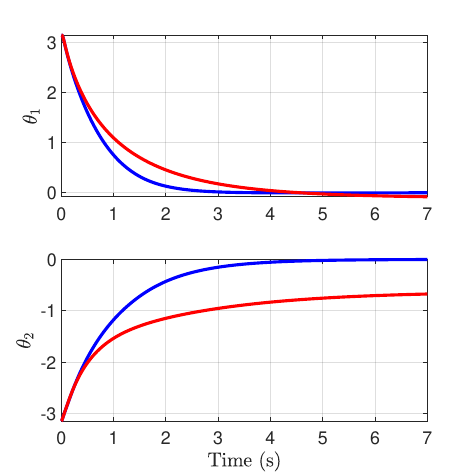}
    \caption{State evolution for the planar manipulator with unknown rotation example. The top plot represents the first revolution angle, whereas the bottom plot represents the second one. The \textcolor{blue}{blue line} corresponds to the adaptive controller, whereas the \textcolor{red}{red line} corresponds to the robust controller.}
    \label{fig:manipulator_states}
\end{figure}

\begin{figure}[htb!]
    \centering
    \includegraphics{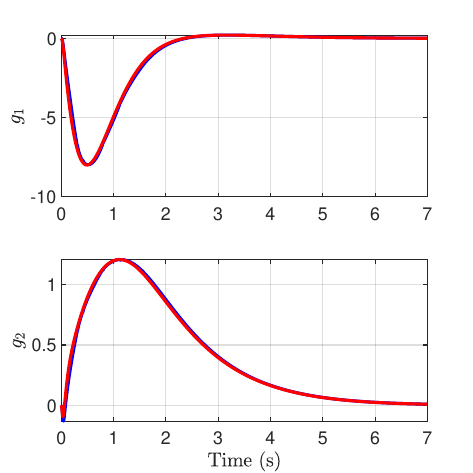}
    \caption{Gravity estimation for the planar manipulator with unknown rotation example. The top plot represents the first component of the gravity vector, $g_1$, whereas the bottom plot represents the second one, $g_2$. The \textcolor{blue}{blue line} corresponds to the estimated gravity effect, whereas the \textcolor{red}{red line} corresponds to the exact one.}
    \label{fig:manipulator_estimations}
\end{figure}

Both controllers were simulated in MATLAB using ode45, with the control inputs held between sample times to simulate their behaviour in closed loop. The unknown function $\vec{\varphi}(\theta_1,\theta_2)$ was approximated using 7 equally distributed triangular membership functions in $\theta_1 \in [-\pi,\pi]$ and 7 equally distributed triangular membership functions in $\theta_2 \in [-\pi, \pi]$, leading to a TS approximator with 49 rules. We are using the fact that this unknown function is periodic in both $\theta_1$ and $\theta_2$, and any state outside of these ranges is simply mapped inside it by adding or subtracting $2\pi$ until it is within the range. Our goal is to drive both $\theta_1$ and $\theta_2$ to zero from an initial condition. In this case, we consider that the initial condition was $\theta_1 = \pi$, $\theta_2 = -\pi$, and the unknown rotations employed in the simulation were $\alpha = -\pi/2$ and $\beta = 3\pi/2$. The proposed local RLS adaptation law was implemented using a 0.95 forgetting factor. These simulations are presented in Figures \ref{fig:manipulator_states} and \ref{fig:manipulator_estimations}. 

As can be clearly seen in Figure \ref{fig:manipulator_estimations}, the proposed local adaptation law was capable of closely approximating the exact effect that gravity has on the manipulator. In addition to this, Figure \ref{fig:manipulator_states} clearly shows that the simple use of a robust controller is not capable of driving the system to the origin in this case, whereas the proposed adaptive controller is capable of driving the system considerably closer to the origin.

\subsection{Sector-bounded nonlinearities}

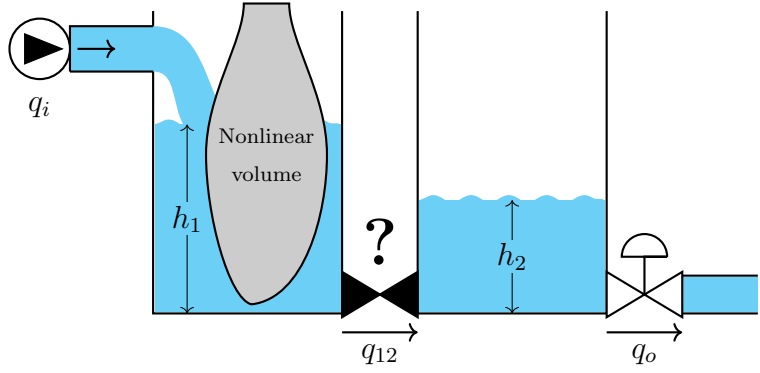
\begin{figure}[htpb]
    \centering
    \begin{tikzpicture}[decoration={snake, segment length=5mm, amplitude=0.5mm}]
        
        \colorlet{watercol}{cyan!50}
        \colorlet{solidcol}{gray!40}

        \fill[watercol] (0.02, 0.02) rectangle (2.48, 2.5);
        \draw[watercol, decorate, line width=2pt] (0.05, 2.48) -- (2.45, 2.48);
        
        \fill[watercol] (3.52, 0.02) rectangle (5.98, 1.5);
        \draw[watercol, decorate, line width=2pt] (3.55, 1.48) -- (5.95, 1.48);

        \fill[watercol] (7.0, 0.02) rectangle (8.0, 0.48);

        \fill[watercol] (-1.1,3.2) rectangle (0,3.8);
        \fill[watercol] (0,3.8) to[in=120,out=0](1,2.5) -- (0.4,2.5) to[in=0,out=120](0,3.2);

        \begin{scope}[shift={(-2.5,0.1)}]
            \draw[thick, fill=solidcol]
                (3.8, 0.02) .. controls (4.5, 0.1) and (4.8, 1.2) .. (4.8, 2.0)
                            .. controls (4.8, 2.8) and (4.3, 3.5) .. (4.3, 4.0) --
                (3.7, 4.0) .. controls (3.7, 3.5) and (3.2, 2.8) .. (3.2, 2.0)
                            .. controls (3.2, 1.2) and (3.5, 0.1) .. (3.8, 0.02);
            \node[align=center] at (4, 2) {\scriptsize Nonlinear \\ \scriptsize volume};
        \end{scope}

        \draw[thick] (0,4) -- (0,3.8);
        \draw[thick] (0,3.2) -- (0,0) -- (2.5,0) -- (2.5,4);
        
        \draw[thick] (3.5,4) -- (3.5,0) -- (6,0) -- (6,4);

        \valvulades{3}{0.25}{$q_{12}$}{\huge \textbf{?}};

        \draw[thick] (7,0) -- (8,0);
        \draw[thick] (7,0.5) -- (8,0.5);
        \valvula{6.5}{0.25}{$q_o$}{};

        \draw[thick] (-1.1,3.2) -- (-1.1,3.8) -- (0,3.8);
        \draw[thick] (-1.1,3.2) -- (0,3.2);

        \draw[thick, fill=white] (-1.5, 3.5) circle (0.4);
        \filldraw[fill=black] (-1.7, 3.3) -- (-1.7, 3.7) -- (-1.2, 3.5) -- cycle;
        \node[below] at (-1.5, 3) {$q_i$};
        \draw[->, thick] (-1.0, 3.5) -- (-0.5, 3.5);

        \draw[<->] (0.45, 0) -- (0.45, 2.5) node[midway, fill=watercol, inner sep=1pt]{$h_1$};
        
        \draw[<->] (4.75, 0) -- (4.75, 1.5) node[midway, fill=watercol, inner sep=1pt]{$h_2$};

    \end{tikzpicture}
    \caption{Two tank system with a nonlinear submerged solid and an unknown coupling between both tanks.}
    \label{fig:dois_tanques}
\end{figure}

Consider the two tank system, with a nonlinear submerged solid and unknown coupling between both tanks illustrated by Figure \ref{fig:dois_tanques}. It can be described by the following equations \citep[simplified from Eq. (19)]{Quadros2020}
\begin{align}
    \dot{h}_1 &= \frac{q_i(u) - q_{12}(h_1,h_2)}{S_1(h_1)}, \\ 
    \dot{h}_2 &= \frac{q_{12}(h_1,h_2) - q_o(h_2)}{S_2},
\end{align}
with
\begin{align}
    S_1(h_1) = \frac{3r}{5}\left(2.7r - \frac{\cos(2.5\pi(h_1-8)\times10^{-2}-\mu)}{\sigma\sqrt{2\pi}}e^{-\frac{((h_1-8)\times10^{-2}-\mu)^2}{2\sigma^2}}\right),
\end{align}
and
\begin{align}
    q_i(u) &= K_b u + q_{i0}, && q_o(h_2) = g_o h_2 + q_{o0}, && q_{12}(h1,h2)= \dfrac{h_1 - h_2}{R_{12}(h_1,h_2)}.
\end{align}
All system parameters are considered to be known, with the exception of the unknown function $R_{12}(h_1,h_2)$ that describes a hydraulic resistance between both tanks, for which only maximum and minimum bounds are known, $R_{\max}$ and $R_{\min}$ respectively.

If we consider a regulation problem around a desired equilibrium point $\begin{bmatrix} h_{10} & h_{20} \end{bmatrix}^T$, with an equilibrium input $u_0$, we can define the deviations $\delta h_1 = h_1 - h_{10}$, $\delta h_2 = h_2 - h_{20}$, and $\delta u = u - u_0$, and the dynamics can be written as
\begin{align}
    \dot{\delta h}_1 &= \dfrac{K_b \delta u - \varphi(\delta h_1,\delta h_2)}{S_1(h_1)} , \\
    \dot{\delta h}_2 &= \dfrac{\varphi(\delta h_1,\delta h_2) -g_0\delta h_2}{S2} ,
\end{align}
with the unknown function $\varphi(\delta h_1,\delta h_2)$ satisfying the sector bounds
\begin{align}
    \begin{bmatrix} \frac{1}{R_{\max}} & \frac{-1}{R_{\max}}\end{bmatrix}\begin{bmatrix}\delta h_1 \\ \delta h_2\end{bmatrix}\leq \varphi(\delta h_1,\delta h_2) \leq \begin{bmatrix} \frac{1}{R_{\min}} & \frac{-1}{R_{\min}}\end{bmatrix}\begin{bmatrix}\delta h_1 \\ \delta h_2\end{bmatrix}.
\end{align}

In order to represent the system in discrete-time, a forward-Euler approach is employed, with a $\Delta t$ sampling time. Though this is an interesting approach to regulating the problem around an equilibrium point, the equilibrium values for the input, $u_0$, and for the level of the first tank, $h_{10}$, cannot be used directly in the control law due to the unknown function $\varphi$. In order to deal with this problem, we augment the system with a discrete-time integral error state, $i_k$, with dynamics
\begin{align}
    i_{k+1} = i_k + \Delta t (h_{2k} - h_{20}),
\end{align}
and also consider the dynamics of its deviation against its steady-state equilibrium value $i_0$, given by $\delta i = i - i_0$.

The deviation dynamics in discrete-time can then be described by
\begin{align}
    \underbrace{\begin{bmatrix} \delta h_{1(k+1)}  \\  \delta h_{2(k+1)} \\ \delta i_{k+1} \end{bmatrix}}_{\delta \vec{x}_{k+1}}\!\!=\!\!\underbrace{\begin{bmatrix} 1 & 0 & 0 \\ 0 & 1 & -\Delta t \frac{g_0}{S_2} \\ 0 & \Delta t & 1 \end{bmatrix}}_{A} \underbrace{\begin{bmatrix} \delta h_{1k}  \\  \delta h_{2k} \\ \delta i_{k} \end{bmatrix}}_{\delta \vec{x}_{k}} \!\!+\!\! \underbrace{\begin{bmatrix}\frac{\Delta t}{S_1(h_1)} \\ 0 \\ 0 \end{bmatrix}}_{B(h_1)} \delta u_k \!+\!\! \underbrace{\begin{bmatrix}-\frac{\Delta t}{S_1(h_1)} \\ \frac{\Delta t}{S_2}\\0\end{bmatrix}}_{G(h_1)} \varphi(\delta h_1, \delta h_2),
\end{align}
which, by using the sector-nonlinearity approach \cite{Tanaka2004} with $\frac{1}{S_1(h_1)}$ as our premise variable, leads to a TS model with 2 vertices.

The addition of the deviation of the integral state to the dynamics solves the problem with the unknown equilibrium point values in the control law, if we force the conditions in Theorem \ref{thm:controle_setor} to fixed $K$ and $F$ gains. In this case, the control law found by employing Theorem \ref{thm:controle_setor} would take the form
\begin{align}
    \delta u_k &= K\delta x_k - F\bar{\varphi}(\delta h_1, \delta h_2), \\
    u_k &= k_1 h_{1k} + k_2 h_{2k} + k_3i_k - F\left(\hat{\varphi}(\delta h_1, \delta h_2) - \frac{h_{1k} - h_{2k}}{R_{\max}}\right) \\&+  \left(u_{0} - k_1 h_{10} - k_2 h_{20} - k_3 i_0 + F\frac{h_{10} - h_{20}}{R_{\max}}\right),
\end{align}
and since $i_0$ is arbitrary, it can be chosen such that 
\begin{align}
    i_0 = \frac{1}{k_3} \left(u_{0} - k_1 h_{10} - k_2 h_{20} + F\frac{h_{10} - h_{20}}{R_{\max}}\right),
\end{align}
and the controller can be implemented as
\begin{align}
    u_k &= k_1 h_{1k} + k_2 h_{2k} + k_3i_k - F\left(\hat{\varphi}(\delta h_1, \delta h_2) - \frac{h_{1k} - h_{2k}}{R_{\max}}\right).
\end{align}

For our numerical experiment, we consider the parameters described in \citep{Quadros2022}, with $S_2 = 0.3019 \;m^2$, $K_b = 15.598 \times 10^{-4} \;m^3/s$, $q_{i0} = 270.34 \times 10^{-6} \;m^3/s$, $g_0 = 12.941 \times 10^{-4}\; m^2/s$, $q_{o0} = 787.586 \times 10^{-6} \; m^3/s$, $r = 0.31\; m$, $\sigma = 0.55 \;m$, and $\mu = 0.40\; m$.

In order to compare the conditions from the proposed adaptive control against a robust control condition (which can be found by forcing $F = 0$ in the conditions in Theorem \ref{thm:controle_setor}), we fix $R_{\max}$ to $600 \; s/m^4$ and search for the smallest value of $R_{\min}$ for which we are able to find a stabilizing controller, with $\Delta t = 1s$, $\eta = 0.02$, and $\lambda = 1 \times 10^{-5}$. For the proposed adaptive controller, we are able to find a feasible solution for Theorem \ref{thm:controle_setor} with $R_{\min} \geq 45$, whereas with the robust controller (obtained by setting $F = 0$ in the condition), we are only able to find a feasible controller for $R_{\min} \geq 71$.

\begin{figure}[htb!]
    \centering
    \includegraphics{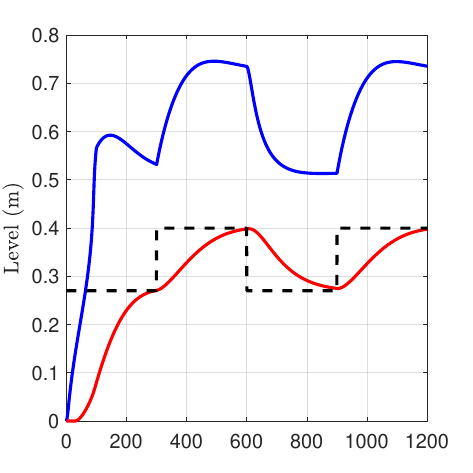}
    \caption{State evolution for the two tanks control example. The \textcolor{blue}{blue line} corresponds to the level of the first tank, the \textcolor{red}{red line} corresponds to the level of the second tank, and the dashed black line corresponds to the level setpoint for the second tank.}
    \label{fig:tanks_states}
\end{figure}

\begin{figure}[htb!]
    \centering
    \includegraphics{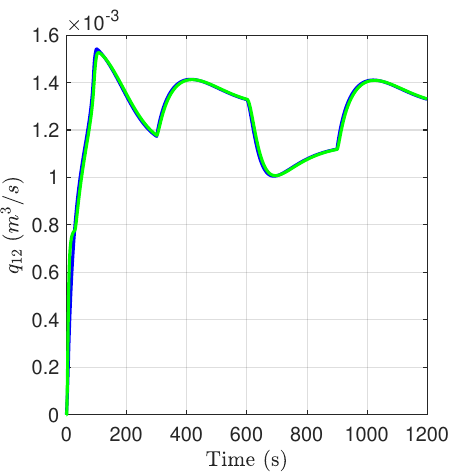}
    \caption{Estimation of the unknown flow between both tanks using the proposed local RLS. The \textcolor{blue}{blue line} corresponds to the real flow, whereas the \textcolor{green}{green line} corresponds to the estimated flow using the TS approximator with the proposed local RLS for the consequents.}
    \label{fig:tanks_flow_estimation}
\end{figure}

The adaptive controller obtained with $R_{\min} = 45$ in Theorem \ref{thm:controle_setor} led to the control gains
\begin{align}
    K = \begin{bmatrix}
        -23.8274 & -49.5449 & -0.7141
    \end{bmatrix}, && F = -481.3354.
\end{align}
In order to illustrate the controller's behaviour in closed loop, we considered $R_{12} = 412(h_1 - h_2) + 114.88$ and simulated it in MATLAB using ode45, with the control inputs held between sampling times. The proposed local RLS approach in \eqref{eq:cov}, \eqref{eq:xi_proj}, and \eqref{eq:theta_proj} with a 0.995 forgetting factor was employed for a TS approximator with 5 gaussian fuzzy sets in $h_1 - h_2 \in [-0.8, 0.8]$ with equally spaced means and $0.5$ standard deviation. The closed loop system behaviour is illustrated in figures \ref{fig:tanks_states} and \ref{fig:tanks_flow_estimation}, showing that the controller is capable of following the desired set-points and the adaptation law is capable of correctly estimating the unknown flow between both tanks.

\subsection{Norm-bounded nonlinearities}

Consider a \emph{Brushless DC} (BLDC) motor, described in the \emph{dq0} coordinates, by the following equations
\begin{align}
    \frac{d}{dt}i_d &= -\frac{R}{L} i_d + p\omega_m i_q+ \frac{1}{L}v_d \\
    \frac{d}{dt}i_q &= -p\omega_m i_d -\frac{R}{L} i_q -\frac{p\psi_m}{L}\omega_m  + \frac{1}{L}v_q \\
    \dot{\omega}_m &= \dfrac{1.5p\psi_mi_q - B_a\omega_m - T_{load}}{J + J_{load}}
\end{align}
with $i_d$ and $i_q$ the direct and quadrature currents, $v_d$ and $v_q$ the direct and quadrature voltages (controlled inputs), $\omega_m$ the motor angular velocity (rad/s), $R$ the motor resistance, $L$ the motor inductance, $\psi_m$ the magnetic flux, $B_a$ the coefficient of the motor viscous friction, $p$ the number of pole pairs, $J$ the motor moment of inertia, and $J_{load}$ and $T_{load}$ unknown terms corresponding to the moment of inertia of the load connected to the motor and the load torque, respectively. The last equation, which defines the angular velocity dynamics, can be rewritten considering the known and unknown terms as
\begin{align}
    \dot{\omega}_m &= \dfrac{T_{known}}{J} - \left(\dfrac{1}{J}\right)\overbrace{\left(\dfrac{J}{J + J_{load}}\right)\left(\dfrac{J_{load}T_{known}}{J} + T_{load}\right)}^{\bar{T}_{load}},
\end{align}
with $T_{known}=1.5p\psi_mi_q - B_a\omega_m$, and $\bar{T}_{load}$ encompassing all the unknown terms and seen as a new unknown load torque by the motor dynamics. This allows the system to be written as
\begin{align}
    \underbrace{\begin{bmatrix} \frac{d}{dt}i_d  \\  \frac{d}{dt}i_q \\ \dot{\omega}_m \end{bmatrix}}_{\dot{\vec{x}}} = \underbrace{\begin{bmatrix} -\frac{R}{L} & p\omega_m & 0 \\ -p\omega_m & -\frac{R}{L} & -\frac{p\psi_m}{L} \\ 0 & \frac{1.5 p\psi_m}{J} & -\frac{B_a}{J}\end{bmatrix}}_{A(\omega_m)} \underbrace{\begin{bmatrix} i_d  \\  i_q \\ \omega_m \end{bmatrix}}_{\vec{x}} + \underbrace{\begin{bmatrix}\frac{1}{L} & 0 \\ 0 & \frac{1}{L} \\ 0 & 0\end{bmatrix}}_{B} \underbrace{\begin{bmatrix}v_d \\ v_q\end{bmatrix}}_{\vec{u}} + \underbrace{\begin{bmatrix}0 \\ 0\\-\frac{1}{J}\end{bmatrix}}_{G} \underbrace{\bar{T}_{load}}_{\varphi(x)},
\end{align}
which, by making use of the sector-nonlinearity approach using $\omega_m$ as our premise variable, leads to a TS model with 2 vertices.

In order to make the problem more interesting, we would like to regulate the motor around a desired angular velocity, $\omega_m^d$, with the maximum conversion of electrical power into mechanical power, thus requiring that the direct current $i_d^d$ be null. If we consider only the known terms in the equation, the desired quadrature current would be given by $i_q^d = \frac{B_a}{1.5p\psi_m}\omega_m^d$ and the direct and quadrature voltages would be given by $v_d^d = -L_p\omega_m i_q^d$ and $v_q^d = Ri_q^d + p\psi_m\omega_m^d$. 

By considering the deviations $\delta\vec{x} = \vec{x} - \vec{x}^d$ and $\delta\vec{u} = \vec{u} - \vec{u}^d$, if the desired angular velocity varies slowly, we can consider that $\dot{\omega}_m^d \approx 0$ and write the dynamics of the deviations as
\begin{align}
    \dot{\delta\vec{x}} = A(\omega_m) \delta\vec{x} + B \delta\vec{u} + G \varphi(\vec{x}).
\end{align}
These can, in turn, be discretized using a Forward-Euler approach as
\begin{align}
    \delta\vec{x}_{k+1} = \underbrace{(I + \Delta t A(\omega_{m_k}))}_{A_d(\omega_{m_k})} \delta\vec{x}_k + \underbrace{\Delta t B}_{B_d} \delta \vec{u}_k + \underbrace{\Delta t G}_{G_d} \varphi(\vec{x}_k). \label{eq:bldc_discreto}
\end{align}

If we consider that $\varphi(\vec{x}_k)$ is bounded, Theorems \ref{thm:controle_norma} and \ref{thm:Francis_regulator} can be employed to develop an adaptive control law that aims to regulate the system around the desired point. In this particular example, we consider that the desired output is given by
\begin{equation}
    \underbrace{\begin{bmatrix} \delta i_{d_k} \\ \delta \omega_{m_k} \end{bmatrix}}_{\vec{y}_k}= \underbrace{\begin{bmatrix} 1 & 0 & 0 \\ 0 & 0 & 1 \end{bmatrix}}_{C} \delta\vec{x}_k.
\end{equation}

For our numerical experiment, we consider that the unknown load corresponds to the load generated by employing the BLDC motor as the thrust of an efficiency vehicle, depicted in Figure \ref{fig:dt1}. In this particular case, the motor generates the thrust force, $F_t$, whereas the vehicle is subjected to both an aerodynamic drag, $F_a$, and a friction term between the wheels and the ground known as rolling resistance, $F_r$.

\begin{figure}
    \centering
    \begin{tikzpicture}[>=Latex]
        \fill [pattern=north east lines, pattern color=black!30] (-1,-0.2) rectangle (10.5,0);
        \draw[thick, black!70] (-1,0) -- (10.5,0);
    
        \draw [thick, fill=gray!10]
            (0.5, 0.5) 
            .. controls (2, 1.2) and (3, 2.2) .. (4.5, 2.3) 
            .. controls (6.5, 2.3) and (8.5, 1.5) .. (8.5, 0.8) 
            .. controls (8.5, 0.4) and (8, 0.3) .. (7.5, 0.3) 
            -- (2, 0.3) 
            .. controls (1, 0.3) and (0.5, 0.4) .. (0.5, 0.5) -- cycle;
    
        \draw [thick, fill=cyan!10]
            (4.2, 1.8)
            .. controls (5.5, 1.9) and (6.8, 1.5) .. (7.2, 1.1)
            .. controls (6.5, 1.1) and (5, 1.1) .. (4.0, 1.4)
            .. controls (3.8, 1.5) and (4.0, 1.7) .. (4.2, 1.8) -- cycle;
    
        \draw[thick, fill=black!80] (2.5, 0.4) circle (0.4); 
        \draw[thick, fill=black!80] (6.5, 0.4) circle (0.4); 
        \draw[thick, fill=gray!40] (2.5, 0.4) circle (0.15); 
        \draw[thick, fill=gray!40] (6.5, 0.4) circle (0.15); 
    
    
        \draw[->, very thick, blue] (3.5, 2.9) -- ++(2.5, 0) node[midway, above] {$\vec{v}$};
    
        \draw[->, very thick, red] (2.5, 0.05) -- ++(2.0, 0) node[above] {$F_t$};
    
        \draw[->, very thick, red] (6.5, 0.4) -- ++(-1.5, 0) node[above] {$F_r$};
    
        \draw[<-, very thick, red] (8.5, 1.2) -- ++(1.5, 0) node[above] {$F_a$};
    
    \end{tikzpicture}
    \caption{Forces acting upon an efficiency vehicle}
    \label{fig:dt1}
\end{figure}
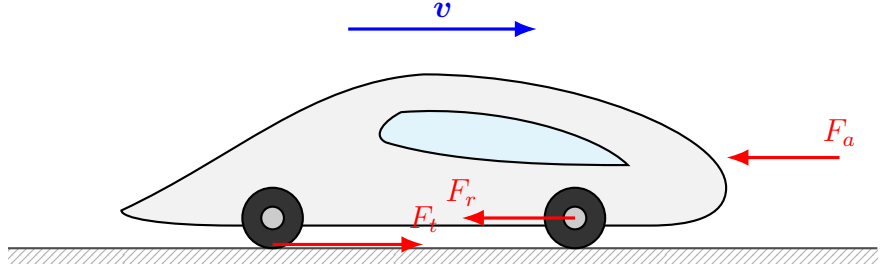

The vehicle speed is given by 
\begin{equation}
    v = \dfrac{R_w}{N} \omega_m,
\end{equation}
with $R_w$ the thrust wheel radius and $N$ the transmission ratio, while the aerodynamic drag and rolling resistance are given by
\begin{align}
    F_a &= \frac{1}{2} \rho C_d A_f v^2, \\
    F_r &= M_t g C_{rr}, 
\end{align}
with $\rho$ the air density, $C_d$ the drag coefficient, $A_f$ the frontal area of the vehicle, $M_t$ the total mass of the vehicle (considering the pilot), $g$ the gravitational acceleration, and $C_{rr}$ the rolling resistance coefficient. The motor sees this as
\begin{align}
    T_{load} &= \left(\dfrac{R_w}{N}\right)(F_a + F_r), \\
    J_{load} &= \left(\dfrac{R_w}{N}\right)^2 M_t + n_wJ_r,
\end{align}
with $n_w$ the number of vehicle wheels, and $J_r$ the moment of inertia of each wheel. 

In this work, we consider a sampling time of $\Delta t = 1\times 10^{-3}$ s and employ the parameters corresponding to the motor and vehicle dynamics of the DT1 prototype from Team \emph{Milhagem UFMG Elétrico}, which won the \emph{Shell EcoMarathon Brasil} in 2018 and 2022. These parameters are summarized in Table \ref{tab:dt1}.

\begin{table}[!ht]
    \centering 
    \begin{tabular}{|c|c|}
    \hline
     Parameter & Value \\ \hline
     Phase resistance, $R$ & 0.035 $\Omega$ \\ \hline
     Phase inductance, $L$ & 0.000225 H \\ \hline
     Rotor magnetic flux, $\psi_m$ & 0.01983 Wb \\ \hline
     Motor moment of inertia, $J$ & $6.52 \times 10^{-5}$ $\text{kg.m}^2$ \\ \hline
     Viscous friction coefficient, $B_a$ & $0.001 \, \text{N.m.s}$ \\ \hline
     Number of pole pairs, $p$ & 4 \\ \hline
     Total mass of the vehicle (vehicle + driver), $M_t$ & 100 kg \\ \hline
     Wheel's moment of inertia, $J_r$ & $0.015 \, \text{kg.m}^2$ \\ \hline
     Number of wheels, $n_w$ & 3 \\ \hline
     Wheel's radius, $R_w$ & 0.254 m \\ \hline
     Transmission ratio, $N$ & 13 \\ \hline
     Frontal area of the vehicle, $A_f$ & $0.26\,\text{m}^2$ \\ \hline
     Drag coefficient, $C_d$ & 0.164 \\ \hline
     Rolling resistance coefficient, $C_{rr}$ & 0.0024 \\ \hline
     Air density, $\rho$ & $1.22\,\text{kg/m}^3$ \\ \hline
     Gravity acceleration, $g$ & $9.81\,\text{m/s}^2$ \\ \hline
    \end{tabular}
    \caption{Motor and vehicle parameters employed in the norm-bounded nonlinearities example}
    \label{tab:dt1}
\end{table}

In order to make use of the conditions in Theorem \ref{thm:controle_norma}, we considered that the motor's angular speed was bounded with $\omega_m \in [-10, 150]$ rad/s and solved for the LMI conditions with $\lambda = 10$, $\eta = 0.4$, and with an objective function $\nu_1 + \nu_2$. This led to the control gains and membership functions
\begin{align}\label{eq:controlgains_BLDC}
    K_1 &= \begin{bmatrix} -0.1900 & 0.0090 & 0 \\ -0.0090 & -0.3908 & -0.0290 \end{bmatrix}, &&\mu_1(\omega_k)=\dfrac{150 -\omega_k}{160}, \\
    K_2 &= \begin{bmatrix} -0.1900 & -0.1350 & 0 \\ 0.1350 & -0.3908 & -0.0290 \end{bmatrix}, &&\mu_2(\omega_k)=\dfrac{\omega_k + 10}{160}.
\end{align}

We considered a scenario in which the motor should try to start from a complete stop, ramp up over 4 seconds to 100 rad/s, and hold that speed until the simulation ends at 16 s. The controllers were simulated in MATLAB using ode45, with the control inputs held between sampling times. Three different control laws were simulated: a simple robust controller using the gains in \eqref{eq:controlgains_BLDC}, an adaptive controller using the local RLS approach in \eqref{eq:cov}, \eqref{eq:xi_proj2}, and \eqref{eq:theta_proj2} with $a = 0.3$, and an adaptive controller using a regular RLS estimator with a forgetting factor of $0.95$. In both estimators, we considered that the unknown nonlinearity $\bar{T}_{load}$ could be approximated with 5 fuzzy sets in the quadrature current, $i_q \in [-100, 100]$, and 5 fuzzy sets in the angular velocity, $\omega_m \in [-10, 150]$, both using equally spaced triangular membership functions, leading to fuzzy approximators with 25 rules in both cases. In addition to this, both estimators considered that $\|\bar{T}_{load}\| \leq 10$.

\begin{figure}[htb!]
    \centering
    \includegraphics{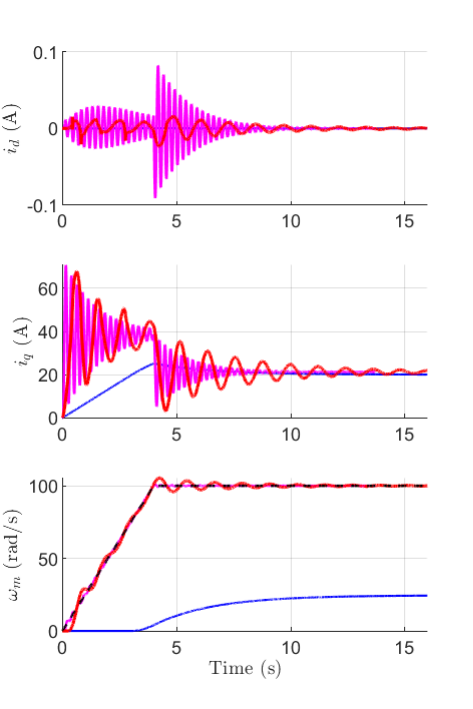}
    \caption{State evolution for the BLDC motor control example. The \textcolor{blue}{blue line} corresponds to the robust controller, the \textcolor{red}{red line} corresponds to the adaptive controller with the local RLS estimator, and the \textcolor{magenta}{magenta line} corresponds to the adaptive controller with a standard RLS estimator with forgetting factor. The top plot is the direct current over time, the middle plot is the quadrature current over time, and the bottom plot is the angular velocity over time.}
    \label{fig:dt1_states}
\end{figure}

\begin{figure}[htb!]
    \centering
    \includegraphics{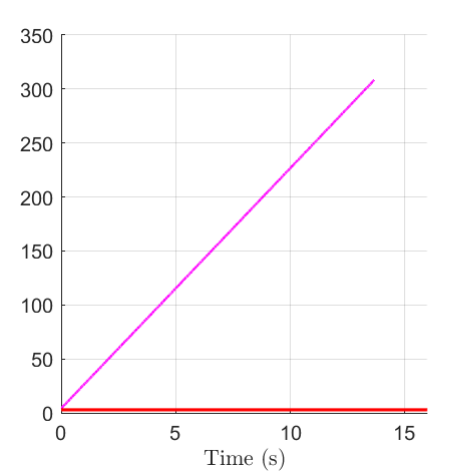}
    \caption{Comparison of the covariance of the different RLS estimators. The \textcolor{red}{red line} corresponds to the $\log_{10}$ of the maximum value of $p_{ik}$ over time, whereas the \textcolor{magenta}{magenta line} corresponds to the $\log_{10}$ of the trace of the covariance matrix over time.}
    \label{fig:dt1_cov}
\end{figure}

\begin{figure}[htb!]
    \centering
    \includegraphics{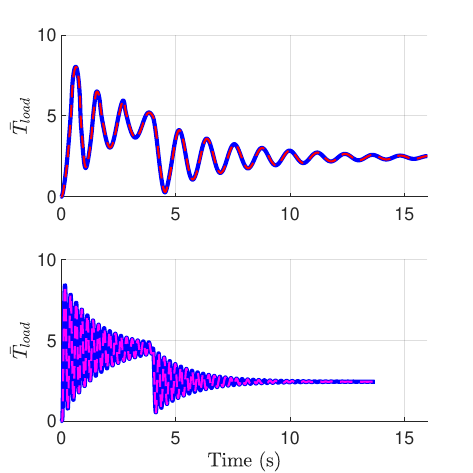}
    \caption{Comparison of the $\bar{T}_{load}$ estimation of the different RLS estimators. The top plot shows the estimation in the scenario using the local RLS proposed in this paper (dashed \textcolor{red}{red} line) against the exact $\bar{T}_{load}$ (solid \textcolor{blue}{blue} line). The bottom plot show the estimation in the scenario using the global RLS with a forgetting factor of 0.95 (dashed \textcolor{magenta}{magenta} line) against the exact $\bar{T}_{load}$ (solid \textcolor{blue}{blue} line).}
    \label{fig:dt1_tload}
\end{figure}

Figures \ref{fig:dt1_states}, \ref{fig:dt1_cov}, and \ref{fig:dt1_tload} present the evolution of the states, $\log_{10}$ of the estimators' covariances, and $\bar{T}_{load}$ estimation over time, respectively. 

From these, we can see that simply using the robust control gains found from Theorem \ref{thm:controle_norma} (using no feedforward from any estimation of the unknown nonlinearity) leads to a considerably worse result (with regards to regulating the angular velocity). Whereas the standard RLS with a forgetting factor might initially seem to be better suited to this problem (since it considers the model globally instead of locally), it suffers from a lack of persistence of excitation (which is quite common in a setting with fuzzy approximators), leading to a serious \emph{covariance wind-up} and eventually resulting in controller divergence/simulation breakdown at around $13.5$s. The proposed local RLS approach, on the other hand, ensures that the closed loop signals remain bounded without requiring a persistence of excitation condition (since locally it is estimating constants, which require no persistence of excitation condition) and has no \emph{covariance wind-up} (since the fuzzy rules only update the covariance when they are active), leading to a superior controller (that trades off estimator precision for robustness).

\section{Conclusion and future works}

This paper proposed novel conditions for the project of discrete-time adaptive fuzzy control laws for nonlinear systems in a qLPV/Takagi-Sugeno form. In these laws, fuzzy TS models are employed to approximate unknown nonlinearities in the system, and a local recursive least squares estimation (based on the TS model structure) that only updates a rule's parameter and covariance while it is active was proposed, together with a stability analysis of this adaptation strategy (Theorem \ref{thm:estimador}). In order to allow for a wider range of systems than the usual ones in the literature (strict feedback or matched nonlinearities), different sets of LMI conditions were proposed to design the adaptive control laws for the cases of matched (Theorem \ref{thm:controle_casado}), sector-bounded (Theorem \ref{thm:controle_setor}) and norm-bounded (Theorems \ref{thm:controle_norma} and \ref{thm:Francis_regulator}) unknown nonlinearities. These are capable of finding, \emph{offline}, control gains that, when employed together with the proposed adaptation law, guarantee ultimately bounded responses in closed loop.

Three examples were presented to illustrate the three different cases for the unknown nonlinearities, as well as different aspects of the proposed controllers. The first example, a planar manipulator on a rotated plane (with an unknown rotation), illustrates that the adaptive controller is capable of obtaining a tighter ultimate bound than a controller that is robust against the nonlinearity (as the approximation error will usually be smaller than the bound on a nonlinearity). The second example, a two tank system with an unknown coupling between the tanks, illustrates that the adaptive control law is capable of accommodating a larger uncertainty than a purely robust controller. And finally, the third example, a BLDC motor being used as the thrust for an efficiency vehicle, illustrates that the proposed local adaptation laws advantage in overcoming the covariance wind-up problem compared to a standard recursive least squares with a forgetting factor.  

Though the proposed adaptation law is capable of adapting the nonlinearity model online recursively, it does so by only changing the parameters of the consequents of the rules. An evolving TS model \citep{Angelov2004} would be capable of adding and removing rules from the fuzzy model rule base while it is evolving, allowing for greater flexibility and possibly better approximation capabilities. As such, an interesting future direction of this research is to adapt it to employ evolving TS models instead of regular TS models. Though there are other LMI-based adaptive controllers using evolving TS models in the literature \citep{Leite2026}, they often focus on the case in which all of the dynamics are unknown and require an LMI condition to be run online. By focusing on certain unknown functions in the systems, we should be capable of designing the controller gains by a set of LMI conditions that are run offline, as is done in this paper.

Another drawback of the proposed controllers and adaptation laws is that they require full state information in order to work. As such, a very interesting future research direction would be to study how to adapt these conditions to the output feedback case by employing state observers.




\section*{Acknowledgments}

This work has been financed by the Coordenação de Aperfeiçoamento de Pessoal de Nível Superior (CAPES) through the Academic Excellence Program (PROEX), Conselho Nacional de Desenvolvimento Científico e Tecnológico (CNPq) - grant number 308597/2023-0, and Fundação de Amparo à Pesquisa do Estado de Minas Gerais (FAPEMIG) - grant number APQ-04246-25.

\bibliographystyle{elsarticle-num} 
\bibliography{references}

@article{Akhtar2004,
  author = {Suhail Akhtar and Ravinder Venugopal and Dennis S. Bernstein *},
  title = {Logarithmic Lyapunov functions for direct adaptive stabilization with normalized adaptive laws},
  journal = {International Journal of Control},
  volume = {77},
  number = {7},
  pages = {630--638},
  year = {2004},
  publisher = {Taylor \& Francis},
  doi = {10.1080/00207170410001699021},
}

@article{Annaswamy2021,
  author = {Anuradha M. Annaswamy and Alexander L. Fradkov},
  title = {A historical perspective of adaptive control and learning},
  journal = {Annual Reviews in Control},
  volume = {52},
  pages = {18-41},
  year = {2021},
  issn = {1367-5788},
  doi = {10.1016/j.arcontrol.2021.10.014},
}

@ARTICLE{Angelov2004,
  author={Angelov, P.P. and Filev, D.P.},
  journal={IEEE Transactions on Systems, Man, and Cybernetics, Part B (Cybernetics)}, 
  title={An approach to online identification of Takagi-Sugeno fuzzy models}, 
  year={2004},
  volume={34},
  number={1},
  pages={484-498},
  doi={10.1109/TSMCB.2003.817053},
}

@ARTICLE{Araujo2021,
  author={Araújo, Rodrigo Farias and Torres, Leonardo A. B. and Palhares, Reinaldo Martinez},
  journal={IEEE Transactions on Systems, Man, and Cybernetics: Systems}, 
  title={Distributed Control of Networked Nonlinear Systems via Interconnected Takagi–Sugeno Fuzzy Systems With Nonlinear Consequent}, 
  year={2021},
  volume={51},
  number={8},
  pages={4858-4867},
  doi={10.1109/TSMC.2019.2945500}
}

@book{Bernstein2009,
  title={Matrix mathematics: theory, facts, and formulas},
  author={Bernstein, Dennis S},
  year={2009},
  publisher={Princeton university press}
}

@ARTICLE{Castro1995,
  author={Castro, J.L.},
  journal={IEEE Transactions on Systems, Man, and Cybernetics}, 
  title={Fuzzy logic controllers are universal approximators}, 
  year={1995},
  volume={25},
  number={4},
  pages={629-635},
  doi={10.1109/21.370193}
}

@article{Coutinho2020,
  author = {Pedro H.S. Coutinho and Rodrigo F. Araújo and Anh-Tu Nguyen and Reinaldo M. Palhares},
  title = {A Multiple-Parameterization Approach for local stabilization of constrained Takagi-Sugeno fuzzy systems with nonlinear consequents},
  journal = {Information Sciences},
  volume = {506},
  pages = {295-307},
  year = {2020},
  issn = {0020-0255},
  doi = {10.1016/j.ins.2019.08.008},
}

@article{Dass2023,
  author = {Anuli Dass and Smriti Srivastava and Rajesh Kumar},
  title = {A novel Lyapunov-stability-based recurrent-fuzzy system for the Identification and adaptive control of nonlinear systems},
  journal = {Applied Soft Computing},
  volume = {137},
  pages = {110161},
  year = {2023},
  issn = {1568-4946},
  doi = {10.1016/j.asoc.2023.110161},
}

@article{Djaneye-Boundjou2025,
  author = {Ouboti Djaneye-Boundjou and Raúl Ordóñez},
  title = {Discrete-time, normalized recursive least squares based concurrent learning for online function approximation},
  journal = {Automatica},
  volume = {177},
  pages = {112273},
  year = {2025},
  issn = {0005-1098},
  doi = {10.1016/j.automatica.2025.112273},
}

@ARTICLE{Dong2009,
  author={Dong, Jiuxiang and Wang, Youyi and Yang, Guang-Hong},
  journal={IEEE Transactions on Systems, Man, and Cybernetics, Part B (Cybernetics)}, 
  title={Control Synthesis of Continuous-Time T-S Fuzzy Systems With Local Nonlinear Models}, 
  year={2009},
  volume={39},
  number={5},
  pages={1245-1258},
  doi={10.1109/TSMCB.2009.2014961}
}

@book{Farrel2006,
  title={Adaptive approximation based control: unifying neural, fuzzy and traditional adaptive approximation approaches},
  author={Farrell, Jay A and Polycarpou, Marios M},
  year={2006},
  publisher={John Wiley \& Sons}
}

@article{Fu2008,
  author = {Simon Hsu-Sheng Fu and Chi-Cheng Cheng},
  title = {On direct adaptive control design for nonlinear discrete-time uncertain systems},
  journal = {Journal of the Franklin Institute},
  volume = {345},
  number = {2},
  pages = {119-135},
  year = {2008},
  issn = {0016-0032},
  doi = {10.1016/j.jfranklin.2007.07.002},
}

@article{Ge2009,
  author = {Shuzhi Sam Ge and Chenguang Yang and Shi-Lu Dai and Zongxia Jiao and Tong Heng Lee},
  title = {Robust adaptive control of a class of nonlinear strict-feedback discrete-time systems with exact output tracking},
  journal = {Automatica},
  volume = {45},
  number = {11},
  pages = {2537-2545},
  year = {2009},
  issn = {0005-1098},
  doi = {10.1016/j.automatica.2009.07.025},
}

@ARTICLE{Goel2020,
  author={Goel, Ankit and Bruce, Adam L. and Bernstein, Dennis S.},
  journal={IEEE Control Systems Magazine}, 
  title={Recursive Least Squares With Variable-Direction Forgetting: Compensating for the Loss of Persistency [Lecture Notes]}, 
  year={2020},
  volume={40},
  number={4},
  pages={80-102},
  doi={10.1109/MCS.2020.2990516},
}

@article{Guo2023,
  author = {Kai Guo and Yongping Pan},
  title = {Composite adaptation and learning for robot control: A survey},
  journal = {Annual Reviews in Control},
  volume = {55},
  pages = {279-290},
  year = {2023},
  issn = {1367-5788},
  doi = {10.1016/j.arcontrol.2022.12.001},
}

@book{Huang2004,
  title={Nonlinear Output Regulation: Theory and Applications},
  author={Huang, J.},
  isbn={9780898715620},
  lccn={2004052533},
  series={Advances in Design and Control},
  year={2004},
  publisher={Society for Industrial and Applied Mathematics}
}

@article{Johanssen1995,
  author = {Tor A. Johansen and Bjarne A. Foss},
  title = {Identification of non-linear system structure and parameters using regime decomposition},
  journal = {Automatica},
  volume = {31},
  number = {2},
  pages = {321-326},
  year = {1995},
  issn = {0005-1098},
  doi = {10.1016/0005-1098(94)00096-2},
}

@article{Leite2026,
  author = {Daniel Leite and Reinaldo Palhares and Igor Škrjanc and Fernando Gomide},
  title = {Evolving granular fuzzy control: Overview, case study on the chaotic Hénon map, and research outlook},
  journal = {Applied Soft Computing},
  volume = {190},
  pages = {114639},
  year = {2026},
  issn = {1568-4946},
  doi = {10.1016/j.asoc.2026.114639},
}

@ARTICLE{Nguyen2019,
  author={Nguyen, Anh-Tu and Taniguchi, Tadanari and Eciolaza, Luka and Campos, Victor and Palhares, Reinaldo and Sugeno, Michio},
  journal={IEEE Computational Intelligence Magazine}, 
  title={Fuzzy Control Systems: Past, Present and Future}, 
  year={2019},
  volume={14},
  number={1},
  pages={56-68},
  doi={10.1109/MCI.2018.2881644},
}

@article{Ortega2020,
  author = {Romeo Ortega and Vladimir Nikiforov and Dmitry Gerasimov},
  title = {On modified parameter estimators for identification and adaptive control. A unified framework and some new schemes},
  journal = {Annual Reviews in Control},
  volume = {50},
  pages = {278-293},
  year = {2020},
  issn = {1367-5788},
  doi = {10.1016/j.arcontrol.2020.06.002},
}

@ARTICLE{Ortega2021,
  author={Ortega, Romeo and Aranovskiy, Stanislav and Pyrkin, Anton A. and Astolfi, Alessandro and Bobtsov, Alexey A.},
  journal={IEEE Transactions on Automatic Control}, 
  title={New Results on Parameter Estimation via Dynamic Regressor Extension and Mixing: Continuous and Discrete-Time Cases}, 
  year={2021},
  volume={66},
  number={5},
  pages={2265-2272},
  doi={10.1109/TAC.2020.3003651},
}

@ARTICLE{Pan2024,
  author={Pan, Yongping and Shi, Tian and Ortega, Romeo},
  journal={IEEE Transactions on Automatic Control}, 
  title={Comparative Analysis of Parameter Convergence for Several Least-Squares Estimation Schemes}, 
  year={2024},
  volume={69},
  number={5},
  pages={3341-3348},
  doi={10.1109/TAC.2023.3326054},
}

@ARTICLE{Qiu2022,
  author={Qiu, Jianbin and Wang, Tong and Sun, Kangkang and Rudas, Imre J. and Gao, Huijun},
  journal={IEEE Transactions on Fuzzy Systems}, 
  title={Disturbance Observer-Based Adaptive Fuzzy Control for Strict-Feedback Nonlinear Systems With Finite-Time Prescribed Performance}, 
  year={2022},
  volume={30},
  number={4},
  pages={1175-1184},
  doi={10.1109/TFUZZ.2021.3053327},
}

@article{Quadros2020,
  author = {Mariella Maia Quadros and Iury Valente de Bessa and Valter J.S. Leite and Reinaldo Martinez Palhares},
  title = {Fault tolerant control for linear parameter varying systems: An improved robust virtual actuator and sensor approach},
  journal = {ISA Transactions},
  volume = {104},
  pages = {356-369},
  year = {2020},
  issn = {0019-0578},
  doi = {10.1016/j.isatra.2020.05.010},
}

@article{Quadros2022,
  author = {Mariella Maia Quadros and Valter J.S. Leite and Reinaldo Martínez Palhares},
  title = {Robust fault hiding approach for T–S fuzzy systems with unmeasured premise variables},
  journal = {Information Sciences},
  volume = {589},
  pages = {690-715},
  year = {2022},
  issn = {0020-0255},
  doi = {10.1016/j.ins.2021.12.129},
}

@article{Sun2024,
  author = {Sun, Aoyu and Zhao, Lin},
  title = {Command filtering-based adaptive backstepping control for a class of discrete-time uncertain nonlinear systems with mismatched disturbances},
  journal = {International Journal of Robust and Nonlinear Control},
  volume = {34},
  number = {17},
  pages = {11513-11528},
  doi = {10.1002/rnc.7586},
  year = {2024},
}

@ARTICLE{Takagi1985,
  author={Takagi, Tomohiro and Sugeno, Michio},
  journal={IEEE Transactions on Systems, Man, and Cybernetics}, 
  title={Fuzzy identification of systems and its applications to modeling and control}, 
  year={1985},
  volume={SMC-15},
  number={1},
  pages={116-132},
  doi={10.1109/TSMC.1985.6313399},
}

@book{Tanaka2004,
  title={Fuzzy control systems design and analysis: a linear matrix inequality approach},
  author={Tanaka, Kazuo and Wang, Hua O},
  year={2004},
  publisher={John Wiley \& Sons}
}

@article{Tao2014,
  author = {Gang Tao},
  title = {Multivariable adaptive control: A survey},
  journal = {Automatica},
  volume = {50},
  number = {11},
  pages = {2737-2764},
  year = {2014},
  issn = {0005-1098},
  doi = {10.1016/j.automatica.2014.10.015},
}

@ARTICLE{Wang1993,
    author={Wang, L.-X.},
    journal={IEEE Transactions on Fuzzy Systems}, 
    title={Stable adaptive fuzzy control of nonlinear systems}, 
    year={1993},
    volume={1},
    number={2},
    pages={146-155},
    doi={10.1109/91.227383}
}

@book{Wang1997,
  title={A Course in Fuzzy Systems and Control},
  author={Wang, L.-X.},
  isbn={9780135408827},
  lccn={96017721},
  year={1997},
  publisher={Prentice Hall PTR}
}

@article{Xiao2015,
  author = {Zhongming Xiao and Tieshan Li and Zifu Li},
  title = {A novel single fuzzy approximation based adaptive control for a class of uncertain strict-feedback discrete-time nonlinear systems},
  journal = {Neurocomputing},
  volume = {167},
  pages = {179-186},
  year = {2015},
  issn = {0925-2312},
  doi = {10.1016/j.neucom.2015.04.079},
}

@article{Yang2015,
  author = {Yang, Xiong and Liu, Derong and Wei, Qinglai and Wang, Ding},
  title = {Direct adaptive control for a class of discrete-time unknown nonaffine nonlinear systems using neural networks},
  journal = {International Journal of Robust and Nonlinear Control},
  volume = {25},
  number = {12},
  pages = {1844-1861},
  doi = {10.1002/rnc.3181},
  year = {2015},
}

@ARTICLE{Zhang2024,
  author={Zhang, Yanqi and Wang, Xin and Wang, Zhenlei},
  journal={IEEE Transactions on Fuzzy Systems}, 
  title={Discrete-Time Adaptive Fuzzy Finite-Time Tracking Control for Uncertain Nonlinear Systems}, 
  year={2024},
  volume={32},
  number={2},
  pages={649-659},
  doi={10.1109/TFUZZ.2023.3303173},
}

\end{document}